\pdfoutput=1
\documentclass{IEEEtran}

\IEEEoverridecommandlockouts
\usepackage{cite}
\usepackage{amsmath,amssymb,amsfonts}
\usepackage{amsthm}
{

	\newtheorem{lemma}{Lemma}
	\newtheorem{theorem}{Theorem}
	
	\newtheorem{remark}{Remark}
}
\usepackage{algorithmic}
\usepackage{graphicx}
\usepackage{textcomp}
\usepackage{xcolor}
\usepackage{url}
\usepackage{subfigure}
\usepackage{float}
\usepackage[top=0.71in, bottom=0.96in, left=0.635in, right=0.635in]{geometry}
\def\BibTeX{{\rm B\kern-.05em{\sc i\kern-.025em b}\kern-.08em
    T\kern-.1667em\lower.7ex\hbox{E}\kern-.125emX}}

\begin{document}

\title{Optimizing Information Freshness for Cooperative IoT Systems with Stochastic Arrivals
}


\author{\IEEEauthorblockN{Bohai Li, Qian Wang, He Chen, Yong Zhou, and Yonghui Li}
	\thanks{The work of H. Chen is supported by the CUHK direct grant under the project code 4055126.}
	\thanks{B. Li and Q. Wang are with School of Electrical
		and Information Engineering, The University of Sydney, Sydney, NSW
		2006, Australia and Department of Information Engineering, The Chinese University of Hong Kong, Hong Kong SAR, China. The work was initiated when B. Li was a visiting student at CUHK (email: \{bohai.li, qian.wang2\}@sydney.edu.au).}%
	\thanks{H. Chen is with Department of Information Engineering, The Chinese University of Hong Kong, Hong Kong SAR, China (email: he.chen@ie.cuhk.edu.hk).}%
	\thanks{Y. Zhou is with School of Information Science and Technology, ShanghaiTech University, Shanghai, 201210, China (email: zhouyong@shanghaitech.edu.cn).}
	\thanks{Y. Li is with School of Electrical
		and Information Engineering, The University of Sydney, Sydney, NSW
		2006, Australia (email: yonghui.li@sydney.edu.au).}
}


\maketitle

\begin{abstract} This paper considers a cooperative Internet of Things (IoT) system with a source aiming to transmit randomly generated status updates to a designated destination as timely as possible under the help of a relay. We adopt a recently proposed concept, the age of information (AoI), to characterize the timeliness of the status updates. In the considered system, delivering the status updates via the one-hop direct link will have a shorter transmission time at the cost of incurring a higher error probability, while the delivery of status updates through the two-hop relay link could be more reliable at the cost of suffering longer transmission time. Thus, it is important to design the relaying protocol of the considered system for optimizing the information freshness. 
Considering the limited capabilities of IoT devices, we propose two low-complexity  age-oriented  relaying (AoR) protocols, i.e., the source-prioritized AoR (SP-AoR) protocol and the relay-prioritized AoR (RP-AoR) protocol, to reduce the AoI of the considered system.  Specifically, in the SP-AoR protocol, the relay opportunistically replaces the source to retransmit the successfully received status updates that have not been correctly delivered to the destination, but the retransmission at the relay can be preempted by the arrival of a new status update at the source. Differently, in the RP-AoR protocol, once the relay replaces the source to retransmit the status updates that have not been successfully received by the destination, the retransmission at the relay will not be preempted by new status update arrivals at the source. By carefully analyzing the evolution of the instantaneous AoI, we derive closed-form expressions of the average AoI for both proposed AoR protocols. We further optimize the generation probability of the status updates at the source in both protocols. Simulation results validate our theoretical analysis, and demonstrate that the two proposed protocols outperform each other under various system parameters. Moreover, the protocol with better performance can achieve near-optimal performance compared with the optimal scheduling policy attained by applying the Markov decision process (MDP) tool.



	

\end{abstract}

\section{Introduction}
With the rapid development of Internet of Things (IoT), timely status updates have become increasingly critical in many emerging IoT applications, such as wireless industrial automation,  autonomous vehicles, and healthcare monitoring \cite{b11}, \cite{b10}. In fact, the conventional performance metrics, e.g., throughput and delay, cannot adequately characterize the timeliness of the status updates \cite{b12}. For example, throughput can be maximized by generating and transmitting the status updates as frequent as possible. However, excessive update rates may lead to severe network congestion, which makes the status updates suffer from long transmission and queuing delays. Such long delays can be reduced by lowering the update rate. On the other hand, if the update rate is reduced too much, the monitor will receive undesired outdated status updates. Motivated by these facts, the age of information (AoI), defined as the time elapsed since the generation of the latest received status update, has been recently introduced to quantify the information freshness from the perspective of the receiver that monitors a remote process \cite{b14}. Unlike the conventional performance metrics, the AoI is related to both the transmission delay and the update generation rate \cite{b13}. As a result, the AoI has attracted increasing attention as a more comprehensive evaluation criterion for information freshness.

\subsection{Background}
Since the AoI concept was first proposed to characterize the information freshness in a vehicular status update system \cite{b1}, extensive studies focusing on the analysis and optimization of the AoI have appeared. Most existing work was concerned with the AoI performance of single-hop wireless networks \cite{b15, b16, b17, b18, b20, b21, b2, b19, b22, b23,b24,b25,b26,b27,b28,b29,b31, b30,b32,b33}. On the other hand,
the AoI performance of multi-hop networks has also been studied in \cite{ b34,b35, b36,  b37, b38,b39}. References \cite{ b34,b35, b36,  b37} focused on the multi-hop networks with a single source. The authors in \cite{b34} considered a general multi-hop network, where a single source disseminates status updates through a gateway to the whole network. They proved that the preemptive Last Generated First Served (LGFS) policy is age-optimal among all causal policies when the packet transmission time is exponentially distributed, and for arbitrary general distribution of packet transmission time, the non-preemptive LGFS policy minimizes the age among all non-preemptive work-conserving policies. In \cite{b35}, the authors considered a simple three-node relay network where the relay
not only forwards packets generated by another stream, but also needs to transmit its own age-sensitive packets.
Under this scenario, a closed-form expression of the average AoI of the relay's packets was provided by leveraging specific queuing theory tools. Reference \cite{b36} characterized the average AoI at the input and output of each node in a line network, where a source delivers status updates to a destination monitor through a series of relay nodes. In \cite{b37}, both the optimal offline and online scheduling policies were proposed to minimize the AoI of a two-hop energy harvesting network. Both \cite{b38} and \cite{b39} focused on the AoI in multi-source, multi-monitor, and multi-hop networks. In \cite{b38}, the AoI of the considered multi-hop networks was studied under general interference constraints, and the optimal stationary policy minimizing the AoI in the studied system was also derived. The authors in \cite{b39} studied the multi-hop networks from a \textit{global} perspective in the sense that every node in the network is both a source and a monitor, and derived the lower bounds for peak and average AoI. An algorithm generating near-optimal periodic status update schedules was derived in \cite{b39}.

%

\subsection{Motivation and Contributions}
All the aforementioned work on multi-hop networks overlooked the direct link between the source and destination. Therefore, the updates from the source can only be transmitted to the destination via the relay. As far as we know, although leveraging the direct link between the source and destination can potentially enhance the information freshness, no existing work has designed the relaying strategy and analyzed the average AoI of a cooperative IoT  system with the existence of a direct link. Such a design is indeed non-trivial. This is because delivering the status updates via the one-hop direct link takes a shorter transmission time at the cost of incurring a potentially higher error probability, while the delivery of status updates through the two-hop relay link could be more reliable at the cost of suffering longer transmission time. By considering the AoI at all three nodes, the optimal design can be attained by applying the Markov design process (MDP) tool. However, considering the limited capabilities of IoT devices \cite{b40}, the 
optimal policy of the MDP problem with three-dimensional state space may not be suitable for most practical IoT systems since it generally has a complex multi-threshold structure. Besides, in practical IoT system design, the closed-form expression of the average AoI is crucial to quickly verify whether the required AoI performance can be guaranteed with a given set of system parameters. The performance of the MDP policy is normally hard to analyze, which makes the verification difficult to be implemented. In this context, two natural questions arise: \textit{For the considered system, are there any policies with simple structures that can achieve near-optimal AoI performance? What is the average AoI performance of such simple policies?} To the best of the authors' knowledge, these questions have not been answered in the literature. Motivated by this gap, in this paper, we investigate a three-node cooperative IoT system, in which a source aims to timely report randomly generated status updates to its destination with the help of a relay. With the existence of a direct link between the source and destination, the transmission of status updates can either go through the one-hop direct link or the two-hop relay link. The goal of this paper is to design, analyze, and optimize simple policies, which can achieve near-optimal AoI performance of the considered system. The main contributions of this paper are summarized as follows:

\begin{itemize}
	\item We first propose two age-oriented relaying (AoR) protocols with simple structures from the perspective of minimizing the AoI in the considered system, namely the source-prioritized AoR (SP-AoR) protocol and the relay-prioritized AoR (RP-AoR) protocol. To ensure that the latest status updates generated at the source can be timely transmitted, we propose the SP-AoR protocol, in which the relay opportunistically replaces the source to retransmit the successfully received status updates that have not been correctly delivered to the destination, but the retransmission at the relay can be preempted by the arrival of a new status update at the source. However, we find that the SP-AoR protocol may not perform well in reducing the AoI in some cases, e.g., the direct link of the considered system suffers from severe channel fading while the generation of new status updates at the source is frequent. Inspired by such cases, we then propose the RP-AoR protocol, in which the retransmission from the relay will not be preempted by a new status update arrival at the source. The two proposed protocols can complement each other in performance under various system parameters.

	\item Based on the AoI evolution process, we define some necessary time intervals to mathematically express the average AoI of both proposed protocols. By representing these time-interval definitions in terms of key system parameters, including the generation probability of the status updates and the transmission success probabilities of three links, we then attain closed-form expressions of the average AoI for both proposed AoR protocols. Note that the analysis is non-trivial compared to the point-to-point case because the AoI of the one-hop direct link and the two-hop relay link is coupled together. In particular, the analysis of the RP-AoR protocol is rather complicated since the possible system state is not unique when a status update is received by the destination through the relay link. Specifically, at the moment of the reception, the considered system may have two possible states: $(a)$ a fresher status update has been generated at the source; $(b)$ there are no fresher status updates at the source. This makes some defined time intervals not independent of each other, resulting in a challenging analysis. 

		
		\item We further minimize the average AoI of both proposed protocols by optimizing the status generation probability at the source. Based on the optimization results, we find that generating status updates as frequently as possible can minimize the average AoI of the RP-AoR protocol. However, in the SP-AoR protocol, generating status updates too frequently in turn increases the average AoI in some cases, e.g., when the channel between the source and destination suffers from severe fading. Simulation results are then provided to validate the theoretical analysis, and demonstrate that the proposed protocols can outperform each other under various  system parameters. Note that based on the analytical results, we can quickly determine which protocol to apply in the considered system to achieve better AoI performance. It is also shown that the protocol with better performance can achieve near-optimal performance compared with the optimal scheduling policy attained by the MDP tool.
		

			
		

\end{itemize}

\subsection{Organization}
The rest of the paper is organized as follows. Section \uppercase\expandafter{\romannumeral2} introduces the system model. We propose, analyze, and optimize the SP-AoR protocol in Section \uppercase\expandafter{\romannumeral3}. Section \uppercase\expandafter{\romannumeral4} provides the design, analysis, and optimization of the RP-AoR protocol. Numerical results are presented in Section \uppercase\expandafter{\romannumeral5} to validate the theoretical analysis and the effectiveness of the proposed AoR protocols. Finally, conclusions are drawn in Section \uppercase\expandafter{\romannumeral6}.

\section{System Model and AoI Definition}
In this section, we describe the system model and present the AoI evolution at each node in the considered system.

\subsection{System Description}
Consider  a three-node cooperative IoT system, in which a source node ($S$) aims to transmit randomly generated status updates to a destination node ($D$) as timely as possible with the help of a relay node ($R$). With the existence of a direct link between $S$ and $D$, the transmission of the status updates can either go through the $S$-$D$ link or the $S$-$R$-$D$ link. We assume that all nodes are equipped with a single antenna and work in the half-duplex mode. Time is divided into slots of equal durations, and the transmission of each status update takes exactly one time slot. We assume that all channels suffer from block fading, i.e., the channel responses remain invariant within one time slot but vary independently from one time slot to another. To quantify the timeliness of the status updates, we adopt a recently proposed AoI metric, first coined in \cite{b1}.

At the beginning of each time slot, the considered cooperative system can have three transmission operations:  $(a)$ $\mathbf{O}_{\mathbf{S}}$: $S$ broadcasts a status update to both $R$ and $D$; $(b)$  $\mathbf{O}_{\mathbf{R}}$: $R$ forwards a status update to $D$; $(c)$ $\mathbf{O}_{\mathbf{N}}$: Neither $S$ nor $R$ transmits a status update. In operation $\mathbf{O}_{\mathbf{S}}$, if $D$ successfully receives the status update from $S$, an acknowledgement (ACK) is sent back to $S$ at the end of the current time slot. If $R$ successfully receives the status update from $S$, in addition to feeding back an ACK to $S$ at the end of the current time slot, it also stores the received status update in its buffer for possible operation $\mathbf{O}_{\mathbf{R}}$ in the following time slot(s). For implementation simplicity, we assume that both $S$ and $R$ can only store one status update, which is discarded when a fresher status update is received successfully. Note that in age-oriented systems, it is meaningless to store stale status updates when a fresher status update is obtained. In operation $\mathbf{O}_{\mathbf{R}}$, $R$ forwards the stored status update to $D$, and $D$ transmits an ACK to $S$ at the end of the current time slot if it successfully receives the status update. The ACK link from both $R$ and $D$ to $S$ is considered to be error-free and delay-free.


\subsection{Age of Information}

Taking into account channel fading, we define $P_{1}$, $P_{2}$, and $P_{3}$ to denote the transmission success probabilities through the $S$-$D$ link, $S$-$R$ link, and $R$-$D$ link, respectively. We follow \cite{b2} and adopt a Bernoulli process to model the stochastic generation of status updates at $S$. Specifically, a new status update is generated with probability $p$ at the beginning of each time slot. Note that the stochastic generation of status updates at $S$ can also be considered as a transmission process through a virtual link with a transmission time of 0. Specifically, we can consider that a transmitter node generates a status update at the beginning of each time slot,  and immediately transmits the status update to $S$ through the virtual link with a transmission success probability of $p$.
 Denote by $U_i(t)$, $i \in \{S, R, D\}$, the generation time of the most recent status update at the receiver side $i$ of a transmission link until time slot $t$. The AoI at receiver $i$ in time slot $t$ can then be defined as 
\begin{equation}
	\Delta_{i}(t)=t-U_{i}(t).
\end{equation}


As the receiver of the virtual link, if $S$ successfully receives the status update from the transmitter node, its local AoI will decrease to 0, otherwise its local AoI will increase by 1. Mathematically, the AoI evolution at $S$ is given by
\begin{equation}\label{aois}
	\Delta_{S}(t+1)=\left\{
	\begin{array}{rcl}
	\Delta_{S}(t)+1,    &      & g(t+1)=0,\\
	0,       &      & g(t+1)=1,
	\end{array} \right.
\end{equation}
where $g(t+1)$ denotes the indicator that is equal to 1 when $S$ successfully receives a status update in the virtual transmission (i.e., a status update is generated at $S$) at the beginning of time slot $t+1$, and $g(t+1)=0$ otherwise. As the receiver of the $S$-$R$ link, if $R$ successfully receives the status update from $S$, it will update the local AoI to match that at $S$, otherwise its local AoI will increase by 1. Mathematically, the AoI evolution at $R$ can be expressed as
\begin{equation}
\Delta_{R}(t+1)=\left\{
\begin{array}{rcl}
\Delta_{S}(t)+1,    &      & w_S(t)=1\ \text{and}\ r_R(t)=1,\\
\Delta_{R}(t)+1,       &      & \text{otherwise},
\end{array} \right.
\end{equation}
where $w_S(t)=1$ denotes that the cooperative system chooses transmission operation $\mathbf{O}_{\mathbf{S}}$ (i.e., $S$ broadcasts a status update to both $R$ and $D$) in time slot $t$, and $r_R(t)=1$ denotes that $R$ successfully receives the status update from $S$ in time slot $t$. Similarly, we can also express the AoI evolution at $D$ as
\begin{equation}\label{aoid}
\Delta_{D}(t+1)=\left\{
\begin{array}{rcl}
\Delta_{S}(t)+1,    &      & w_S(t)=1\ \text{and}\ r_D(t)=1,\\
\Delta_{R}(t)+1,    &      &   w_R(t)=1,\ r_D(t)=1,\\
& &\ \text{and}\ \Delta_{R}(t) < \Delta_{D}(t),\\
\Delta_{D}(t)+1,       &      & \text{otherwise},
\end{array} \right.
\end{equation}
where $w_R(t)=1$ denotes that the cooperative system transmits with operation $\mathbf{O}_{\mathbf{R}}$ (i.e., $R$ forwards the stored status update to $D$) in time slot $t$, and $r_D(t)=1$ denotes that $D$ successfully receives a status update in time slot $t$. The first case in \eqref{aoid} corresponds to the case where $S$ broadcasts a status update and $D$ successfully receives it. In this case, the AoI at $D$ will be updated to match that at $S$. The second case in \eqref{aoid} corresponds to the case where the considered cooperative system transmits with operation $\mathbf{O}_{\mathbf{R}}$ and the stored status update at $R$ is successfully forwarded to $D$. In this case, the AoI at $D$ will be updated to match that at $R$ if the status update from $R$ is fresher. When no update is received at $D$ or the update is staler than the current status at $D$, i.e., the third case in \eqref{aoid}, the AoI will simply increase by 1. Note that, in the considered system, we focus on the long-term average AoI at $D$, which is given by
\begin{equation}
\bar{\Delta}_{D}=\lim\limits_{T \to \infty}\sup\frac{1}{T}\sum_{t=1}^{T}\Delta_{D}(t).
\end{equation}

 
With the help of the ACK, $S$ can have a table of the AoI evolution at each node in the considered system. As such, we assume that the nodes in the system are scheduled in a centralized manner and $S$ serves as the system coordinator. Specifically, based on the ACK received at the end of the current time slot, $S$ has the knowledge of the local AoI at $R$ and $D$ in the next time slot. At the beginning of the next time slot, after updating its own local AoI in the table, $S$ will determine the transmission operation of the system and immediately notify $R$ of the decision. For simplicity, we ignore the time for updating and making decisions at $S$, and the time for notifying $R$ of the decisions.

Note that, in the considered system, we always have  $\Delta_{S}(t)\leq\Delta_{R}(t)$ and $\Delta_{S}(t)\leq \Delta_{D}(t)$. This is because a new status update is first generated at $S$, and then broadcast to $R$ and $D$. Thus, the status updates at $R$ and $D$ cannot be fresher than that at $S$. It is obvious that, if there is $\Delta_{S}(t)=\Delta_{D}(t)$ in the AoI table, i.e., there are no new status updates to be transmitted in the system, $S$ coordinates the system to select $\mathbf{O}_{\mathbf{N}}$. Otherwise, either $\mathbf{O}_{\mathbf{S}}$ or $\mathbf{O}_{\mathbf{R}}$ is chosen to transmit the new status update to reduce the AoI at $D$. Specifically, if there is $\Delta_{S}(t)=\Delta_{R}(t)<\Delta_{D}(t)$ in the table, i.e., $S$ and $R$ have the same new status update to transmit, $S$ coordinates the system to transmit with $\mathbf{O}_{\mathbf{R}}$. Note that the channel gain of the $R$-$D$ link is generally better than that of the $S$-$D$ link. Retransmitting the same status update by $R$ instead of $S$ has a greater probability to reduce the AoI. If there is $\Delta_{S}(t)<\Delta_{R}(t)$ and $\Delta_{R}(t)>\Delta_{D}(t)$ in the table, the system obviously chooses $\mathbf{O}_{\mathbf{S}}$ for transmission. However, if there is $\Delta_{S}(t)<\Delta_{R}(t)<\Delta_{D}(t)$, one natural question arises: \textit{Which of $\mathbf{O}_{\mathbf{S}}$ and $\mathbf{O}_{\mathbf{R}}$ should be chosen to reduce the AoI at $D$? Operating in  $\mathbf{O}_{\mathbf{S}}$ could lead to a fresher status update while operating in  $\mathbf{O}_{\mathbf{R}}$ results in a higher transmission success probability.} The optimal answer to the question can be obtained with the help of the MDP tool. For the MDP problem with three-dimensional state space, the optimal policy generally has a complex multi-threshold structure, which makes the performance of the policy hard to analyze. This difficulty also exists in our considered three-node cooperative IoT system \cite{b42}. Considering the limited capabilities of IoT devices and the requirements for average AoI performance in certain practical IoT applications, we propose two low-complexity protocols for  the considered system and derive their average AoI performance in closed form in the following sections.

\vspace{-0.12cm}
\section{Source-Prioritized AoR Protocols}\label{3}

In this section, we propose a low-complexity SP-AoR protocol to reduce the AoI of the considered system. By analyzing the evolution of the AoI, we manage to attain the closed-form expression of the average AoI for the protocol. Given reasonable value sets of $P_1$, $P_2$, and $P_3$, we further minimize the average AoI by optimizing the status generation probability $p$. 
\vspace{-0.12cm}
\subsection{Protocol Description}

In conventional cooperation protocols, which are designed mainly from a physical layer perspective, dedicated channel resources are normally allocated for $R$ to facilitate the cooperation \cite{b5}. However, if the transmission success probability of the $S$-$D$ link is relatively large, it is apparent that transmitting more status updates through the $S$-$D$ link reduces the AoI. Therefore, conventional cooperation protocols may not perform well from the perspective of minimizing the AoI. We note that a network-level cooperation protocol, where $R$ utilizes the silence periods of $S$ terminals to execute cooperation, was proposed in \cite{b6} to improve the cooperative system performance. Inspired by such a protocol, to ensure that fresh status updates can be transmitted in a timely manner in the SP-AoR protocol, we enforce $S$ to preempt the retransmission from $R$ in case of a new status update arrival. Hence, even if $R$ has a status update to forward at the beginning of one time slot, it still needs to wait for the instruction from $S$. If $S$ has a new status update to transmit, it informs $R$ that the system transmits with $\mathbf{O}_{\mathbf{S}}$ in the current time slot. $R$ then discards its current status update and attempts to decode the new one sent by $S$. Otherwise, $S$ coordinates the system to keep retransmitting the update at $R$ with $\mathbf{O}_{\mathbf{R}}$ until the successful reception at $D$ or the preemption by $S$. It is worth mentioning that the design in \cite{b6} focused on the throughput maximization rather than the AoI minimization. To the best of the authors' knowledge, the SP-AoR protocol is the first effort towards the design and analysis of the considered cooperative system from the perspective of the AoI.
\vspace{-0.13cm}
\subsection{Analysis of Average AoI}\label{3B}
For the ease of understanding the AoI evolution of the SP-AoR protocol, we illustrate an example staircase path for 10 consecutive time slots with an initial value of one in Fig. 1. We denote by $t_k$ the generation time of the $k$th status update received by $D$ and denote by $t_k'$ its arrival time at $D$. Moreover, we denote by $t_*$ the generation time of the discarded status updates in the considered system. 


\begin{figure}
	\centering
	\includegraphics[width=0.9\linewidth]{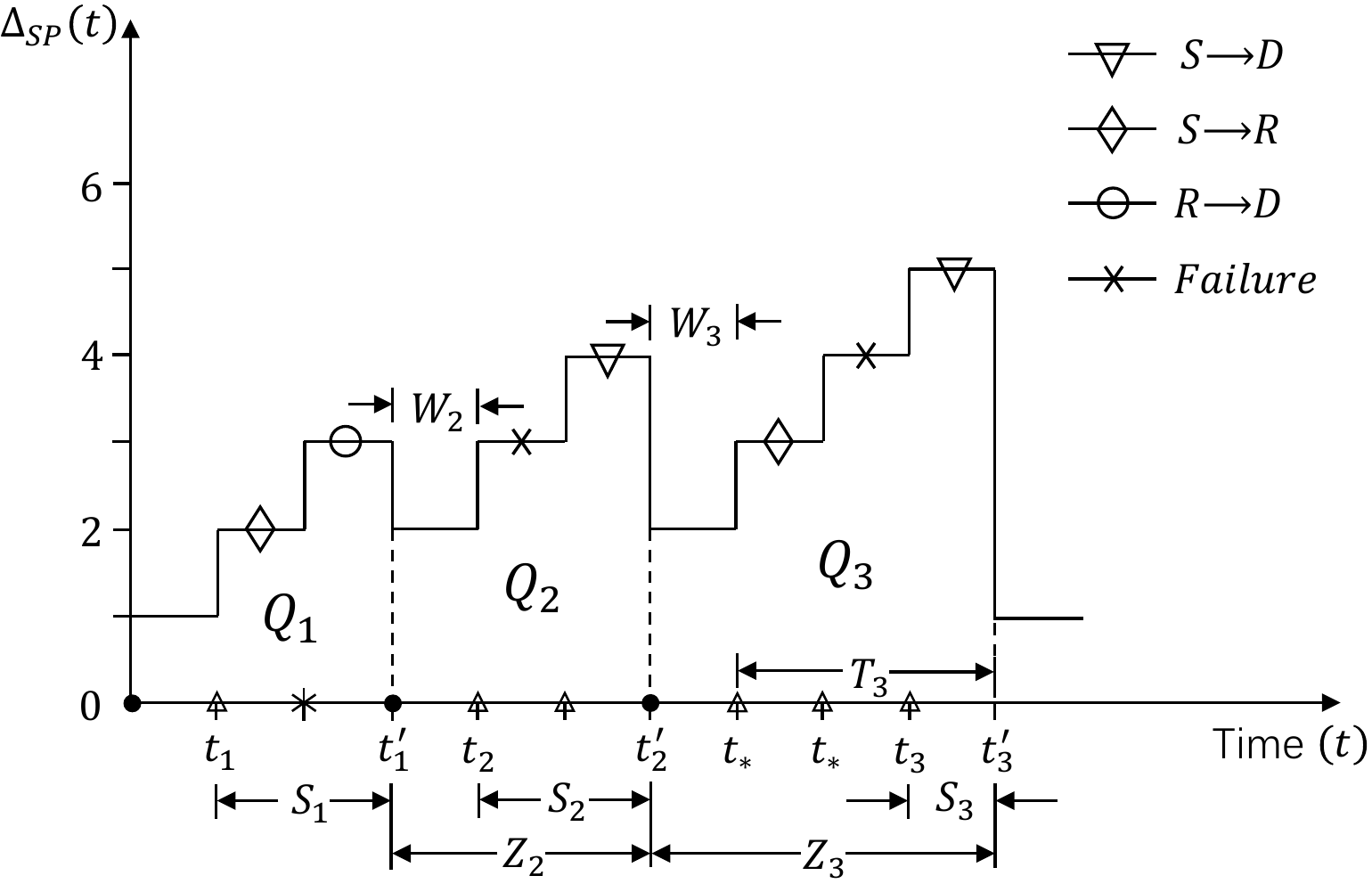}
	\vspace{-0.1cm}
	\caption{Sample staircase path of the AoI evolution for the SP-AoR protocol. We use $S\!\rightarrow\! D$, $S\!\rightarrow\! R$, $R\!\rightarrow\! D$, and Failure to denote the successful transmission through the $S$-$D$ link, the successful transmission through the $S$-$R$ link, the successful transmission through the $R$-$D$ link, and the failed transmission, respectively. Moreover, we use $\bullet$, {\scriptsize{$\triangle$}}, and $*$ to denote the events that none of the links are active, the $S$-$R$ link and the $S$-$D$ link are active, and the $R$-$D$ link is active, respectively.}
	\label{fig:aoisp}
	\vspace{-0.4cm}
\end{figure}

Note that we also define some time intervals in Fig. 1 to facilitate the calculation of the AoI. We define $S_k$ as the service time of the $k$th status update received by $D$, which is given by $S_k=t_k'-t_k$. $T_k$ is defined as the time duration between the generation time of the first status update after $t'_{k-1}$ and the arrival time of the $k$th received update at $D$ (i.e., $t'_k$). In addition, the interdeparture time between two consecutive arrivals of status updates at $D$ is given by $Z_k=t'_k-t'_{k-1}$ and $Z_k=T_k+W_k$, with $W_k$ being the waiting time from the arrival of the $(k-1)$th received update at $D$ (i.e., $t'_{k-1}$) to the generation of the first status update at $S$ after $t'_{k-1}$.


According to the renewal process theory \cite{b43}, the observed time slots $[1,T]$ can be viewed as $X_T$ renewal periods. As depicted in Fig. 1, we denote by $Q_k$ the area under the AoI curve of the $k$th renewal period. Therefore, the average AoI of the SP-AoR protocol can be expressed by
\begin{equation}\label{aoisp1}
	\bar{\Delta}_{SP}=\lim\limits_{T \to \infty }{\frac{X_{T}}{T}\frac{1}{X_{T}}\sum_{k=1}^{X_{T}}Q_{k}}=\frac{\mathbb{E}[Q_{k}]}{\mathbb{E}[Z_{k}]}.
\end{equation}For further simplification, we represent $Q_{k}$ in terms of the time intervals defined above and it can be expressed as
\begin{equation}\label{qksp}
\begin{aligned}
Q_{k}&=S_{k-1}+(S_{k-1}+1)+\cdots+(S_{k-1}+Z_{k}-1)\\
&=S_{k-1}Z_{k}+\frac{Z_{k}^2-Z_{k}}{2}.
\end{aligned}
\end{equation}

\noindent By substituting \eqref{qksp} into \eqref{aoisp1}, we can express the average AoI as 	
\begin{equation}\label{aoisp2}
\bar{\Delta}_{SP}=\frac{\mathbb{E}[S_{k-1}Z_{k}]}{\mathbb{E}[Z_{k}]}+\frac{\mathbb{E}\big[Z_k^2\big]}{2\mathbb{E}[Z_{k}]}-\frac{1}{2}.
\end{equation}
In order to further simplify \eqref{aoisp2}, we have the following lemma.

\begin{lemma} \label{lemma1}
	{In the SP-AoR protocol, the service time of the $(k-1)$th   received update at $D$ is independent of the interdeparture time between the $(k-1)$th and the $k$th received update at $D$ such that
		\begin{equation}
		\mathbb{E}[S_{k-1}Z_{k}]=\mathbb{E}[S_{k-1}]\cdot \mathbb{E}[Z_{k}].
		\end{equation}}
\end{lemma}

\begin{proof}
	As we have $Z_k=W_k+T_k$, Lemma \ref{lemma1} can be proved if the service time of the $(k-1)$th received update at $D$, i.e, $S_{k-1}$,  is independent of both $W_k$ and $T_k$. It is straightforward to find that $S_{k-1}$ is independent of $T_k$ since the two random variables are defined for different updates. Recall that in the SP-AoR protocol, $S$ can preempt the transmission of $R$ when there is a new status update arrival. This indicates that the waiting time $W_k$ only depends on the generation probability $p$. This completes the proof.
\end{proof}

\noindent By applying Lemma 1, the expression of the average AoI can be simplified as
\vspace{-0.1cm}
\begin{equation}\label{aoisp3}
\bar{\Delta}_{SP}=\mathbb{E}[S_{k-1}]+\frac{\mathbb{E}\big[Z_{k}^2\big]}{2\mathbb{E}[Z_{k}]}-\frac{1}{2}.
\end{equation}
To obtain the average AoI of the SP-AoR protocol, we now derive the terms $\mathbb{E}[S_{k-1}]$, $\mathbb{E}[Z_{k}]$ and $\mathbb{E}\big[Z_{k}^2\big]$ one by one in the following. 

\subsubsection{First Moment of the Service Time $\mathbb{E}[S_{k-1}]$}


In the SP-AoR protocol, there are two types of status updates that can be successfully received by $D$ without being preempted, i.e., the updates successfully received by $D$ through the direct link and the updates successfully delivered to $D$ through the two-hop relay link.
Therefore, $\mathbb{E}[S_{k-1}]$ is a weighted sum of the average service time of these two types. For the first type of successful updates, the probability that an update is successfully received after $l$ times of transmissions by $S$ can be expressed as  
\vspace{-0.05cm}	
\begin{equation}\label{PL}
P_{l} = (1-p)^{l-1}(1-P_{2})^{l-1}(1-P_{1})^{l-1}P_{1}.	
\vspace{-0.05cm}
\end{equation}
The corresponding first moment of the service time can be calculated by
\begin{equation}\label{EL}
		\mathbb{E}[l|\mathfrak{D}]=\frac{\sum_{l=1}^{\infty}P_{l}\cdot l}{\sum_{l=1}^{\infty}P_{l}},
\end{equation}
where event $\mathfrak{D}$ denotes that the status updates successfully received by $D$ are transmitted through the direct link without preemption. Similarly, for the second type of successful status updates, the first moment of the service time can be expressed as
\vspace{-0.05cm}
\begin{equation}\label{EMN}
		\mathbb{E}[m+n|\mathfrak{R}]=\frac{\sum_{n=1}^{\infty}\sum_{m=1}^{\infty}P_{mn}\cdot(m+n)}{\sum_{n=1}^{\infty}\sum_{m=1}^{\infty}P_{mn}},
\end{equation} 
\noindent where event $\mathfrak{R}$ denotes that the status updates successfully received by $D$ are transmitted through the two-hop relay link without preemption. $m$ and $n$ denote the number of transmission times via the $S$-$R$ link and the $R$-$D$ link, respectively. $P_{mn}$ is the probability that an update is successfully received by $D$ after transmitting $(m+n)$ times, which can be expressed as
\vspace{-0.025cm}
\begin{equation}\label{PMN}
\resizebox{.98\hsize}{!}{$P_{mn} \!=\! (1-p)^{m-1}(1-P_{2})^{m-1}(1-P_{1})^{m}P_{2}(1-p)^{n}(1-P_{3})^{n-1}P_{3}$}.
\end{equation}
Since the above two types of updates make up all the updates that can be received by $D$ without being preempted, the first moment of the service time for the SP-AoR protocol can be evaluated as 
\begin{equation}\label{Ssp1}
\begin{aligned}
\mathbb{E}[S_{k-1}]\!=\!&\ \mathbb{E}[l|\mathfrak{D}]\!\cdot\!\frac{\sum_{l=1}^{\infty}P_{l}}{\sum_{l=1}^{\infty}P_{l}\!+\!\sum_{n=1}^{\infty}\!\sum_{m=1}^{\infty}\!P_{mn}}\\
&\,\!+\! \mathbb{E}[m\!+\!n|\mathfrak{R}]\!\cdot\!\frac{\sum_{n=1}^{\infty}\!\sum_{m=1}^{\infty}\!P_{mn}}{\sum_{l=1}^{\infty}P_{l}\!+\!\sum_{n=1}^{\infty}\!\sum_{m=1}^{\infty}\!P_{mn}},
\end{aligned}
\end{equation}
\noindent which follows according to the law of total probability.

\begin{figure*}
	\begin{subequations}\label{finite}
		\begin{equation}
			\begin{aligned}
				\sum_{k'=1}^{K}aq^{k'-1}=\frac{a(q^K-1)}{q-1},\ {\rm{with}} \ q\neq 1,
			\end{aligned}
		\end{equation}
	and
	\begin{equation}
		\sum_{k'=0}^{K-1}(a+k'r)q^{k'}=\frac{a-[a+(K-1)r]q^K}{1-q}+\frac{rq(1-q^{K-1})}{(1-q)^2},\  {\rm{with}}\ q\neq1 \  {\rm{and}} \  K>1.
	\end{equation}
	\end{subequations}	
\vspace{-0.2cm}
	\noindent\rule[0.05\baselineskip]{\textwidth}{0.5pt}
\end{figure*}

By applying the finite sum equations given in \eqref{finite} at the top of this page \cite[Eqs. (0.112) and (0.113)]{b7}, we can further simplify \eqref{Ssp1} to
\begin{equation}\label{Ssp2}
\mathbb{E}[S_{k-1}]=\frac{1}{1-\beta}+\frac{1}{1-\alpha}\cdot\frac{\gamma}{P_{1}(1-\alpha)+\gamma},
\end{equation}
where $\alpha=(1-p)(1-P_{3})$, $\beta=(1-p)(1-P_{1})(1-P_{2})$ and $\gamma=P_{2}P_{3}(1-p)(1-P_{1})$.

\subsubsection{First Moment of the Interdeparture Time $\mathbb{E}[Z_{k}]$}

\begin{figure*}
	\begin{equation}\label{Tsp1}
	\begin{aligned}
	\mathbb{E}[T_{k}]=&\
	\sum_{l=1}^{\infty}P_{l}\cdot l + \sum_{n=1}^{\infty}\sum_{m=1}^{\infty}P_{mn}\cdot (m+n)
	+ \sum_{l=1}^{\infty}(1-P_{1})^{l}(1-P_{2})^{l}(1-p)^{l-1}p\cdot \big(l+\mathbb{E}\big[T'_{k}\big]\big)\\ &\,+\sum_{n=0}^{\infty}\sum_{m=1}^{\infty}\Big[(1-P_{1})^{m}(1-P_{2})^{m-1}P_{2}(1-p)^{m-1}
	\times(1-p)^{n}(1-P_{3})^{n}p\cdot \big(m+n+\mathbb{E}\big[T'_{k}\big]\big)\!\Big].\\
	\end{aligned}
	\end{equation}
	\vspace{-0.5cm}
	\noindent\rule[0.05\baselineskip]{\textwidth}{0.5pt}
\end{figure*}

As we have $Z_{k}=W_{k}+T_{k}$, the first moment of the interdeparture time can be written as $\mathbb{E}[Z_{k}]=\mathbb{E}[W_{k}]+\mathbb{E}[T_{k}]$. Recall that, in the SP-AoR protocol, the waiting time $W_{k}$ only depends on the generation probability $p$. Since the status updates are generated according to a Bernoulli process, $W_{k}$ follows a geometric distribution with parameter $p$, and its first moment can be readily given by $\mathbb{E}[W_{k}]=(1-p)/p$. We then move on to the calculation of the term $\mathbb{E}[T_{k}]$.

Note that in the SP-AoR protocol, $T_{k}$ behaves differently for the following four possible cases: $(a)$ The update is successfully received by $D$ through the $S$-$D$ link without being preempted; $(b)$ The update is successfully received by $D$ through the $S$-$R$-$D$ link without being preempted; $(c)$ The update is preempted by a new update before being successfully received by either $R$ or $D$, and the new update may be preempted by multiple new updates; $(d)$ The update is preempted by a new update generated at $S$ after being successfully received by $R$, and the new update may be preempted by multiple new updates at $S$. We note that the number of preemptions can approach infinity in the third and fourth cases, which generally makes the first moment of $T_{k}$ difficult to derive mathematically. To tackle this issue, we resort to the recursive method applied in \cite{b8} and evaluate the first moment of $T_{k}$ as \eqref{Tsp1} at the top of the next page.

Note that the four terms on the right hand side of \eqref{Tsp1} correspond to the above four cases, respectively, and $T'_{k}$ is the time duration between the preemption time (i.e., the generation time of the second status update after $t'_{k-1}$) and the arrival time of the $k$th received update at $D$. As $T_{k}$ is defined as the time duration between the generation time of the first status update after $t'_{k-1}$ and the arrival time of the $k$th received update at $D$, following the idea of recursion, we have $\mathbb{E}[T_{k}]=\mathbb{E}\big[T'_{k}\big]$. After some manipulations by applying \cite[Eqs. (0.112) and (0.113)]{b7}, we have
\begin{equation}\label{Tsp2}
\mathbb{E}[T_{k}]=\frac{(1-\alpha)+\beta\cdot P'_{2}}{P_{1}(1-\alpha)+\gamma},
\end{equation}
where $P'_{2} = P_{2}/(1-P_{2})$. 
By combining $\mathbb{E}[W_{k}]=(1-p)/p$ and \eqref{Tsp2}, we now attain a closed-form expression of the term $\mathbb{E}[Z_{k}]$, given by
\begin{equation}\label{Zsp}
\begin{aligned}
\mathbb{E}[Z_{k}]=\mathbb{E}[W_{k}]+\mathbb{E}[T_{k}]=\frac{(1-\alpha)(1-\beta)}{p\big[P_{1}(1-\alpha)+\gamma\big]}.
\end{aligned}
\end{equation}

\subsubsection{Second Moment of the Interdeparture Time $\mathbb{E}\big[Z_{k}^{2}\big]$}The second moment of the interdeparture time $Z_k$ can be expressed as 
\begin{equation}\label{Z^2sp}
\mathbb{E}\big[Z_k^2\big]\!\!=\!\mathbb{E}\big[(W_{k}+T_{k})^2\big]\!\!=\!\mathbb{E}\big[W_{k}^{2}\big]+2\mathbb{E}[W_{k} T_{k}]+\mathbb{E}\big[T_{k}^{2}\big].
\end{equation}
It is readily to find that $W_{k}$ and $T_{k}$ are independent. Thus, we have $\mathbb{E}[W_{k}T_{k}]=\mathbb{E}[W_{k}]\cdot \mathbb{E}[T_{k}]$. As $W_{k}$ follows a geometric distribution with parameter $p$, we have $\mathbb{E}\big[W_{k}^2\big]=(p^2-3p+2)/p^2$. In order to evaluate $\mathbb{E}\big[Z_{k}^2\big]$, the only remaining task is to calculate $\mathbb{E}\big[T_{k}^2\big]$.

\begin{figure*}
	\begin{equation}\label{T^2sp1}
	\resizebox{.99\hsize}{!}{$\begin{aligned}
		\mathbb{E}\big[T_{k}^{2}\big]\!=\!\!&\
		\sum_{l=1}^{\infty}P_{l}\!\cdot\! l^2 \!+\! \sum_{n=1}^{\infty}\!\sum_{m=1}^{\infty}\!P_{mn}\!\cdot\! \big(m^2\!+\!2mn\!+\!n^2\big)
		\!+\!\sum_{l=1}^{\infty}\!\bigg[\!(1\!-\!P_{1})^{l}(1\!-\!P_{2})^{l}(1\!-\!p)^{l-1}p
		\!\cdot\!\Big(l^2\!+\!2l\!\cdot\!\mathbb{E}\big[T_{k}\big]\!+\!\mathbb{E}\big[T_{k}^{2}\big]\Big)\!\bigg]\\
		&\,\!+\! \sum_{n=0}^{\infty}\!\sum_{m=1}^{\infty}\!\bigg[\!(1\!-\!P_{1})^{m}(1\!-\!P_{2})^{m-1}P_{2}(1\!-\!p)^{m-1}
		(1\!-\!p)^{n}(1\!-\!P_{3})^{n}p
		\!\cdot\! \Big(\!m^2\!+\!n^2\!+\!\mathbb{E}\big[T_{k}^{2}\big]\!
		+\!2mn\!+\!2m\!\cdot\!\mathbb{E}[T_{k}]\!+\!2n\!\cdot\!\mathbb{E}[T_{k}]\!\Big)\!\bigg].
		\end{aligned}$}
	\vspace{-0.05cm}
	\end{equation}
	\vspace{-0.45cm}	
	\noindent\rule[0\baselineskip]{\textwidth}{0.4pt}
\end{figure*}

Similar to \eqref{Tsp1}, the second moment of $T_k$ can be evaluated as \eqref{T^2sp1} at the top of the next page. Similar to the process of obtaining \eqref{Tsp2} from \eqref{Tsp1}, we can rewrite $\mathbb{E}\big[T_{k}^2\big]$ as
	\begin{equation}\label{T^2sp2}
\resizebox{.98\hsize}{!}{$\begin{aligned}
	\mathbb{E}\big[T_{k}^2\big]\!= &\,\frac{1}{(1-\alpha)(1-\beta)-(1-\alpha)\beta\cdot p'-\beta\cdot p'\cdot P'_{2}}\\
	&\times\!\frac{1}{(1-\alpha)(1-\beta)}\\
	&\times\!\bigg\{\!(1-\alpha)^2(1+\beta)+(3-\alpha-\beta-\alpha\beta)\beta\cdot P'_{2}\\
	&\qquad\!\!+2\mathbb{E}[T_{k}]\!\cdot\!\Big[(1\!-\!\alpha)^2\beta\!\cdot \!p'\!+\!(1\!-\!\alpha\beta)\beta\!\cdot\! p'\!\cdot\! P'_{2} \Big]\!\bigg\},
	\end{aligned}$}
\end{equation}
where $p'=p/(1-p)$.



\subsubsection{Average AoI}
By substituting the terms derived in \eqref{Ssp2}, \eqref{Tsp2}, \eqref{Zsp}, \eqref{Z^2sp} and \eqref{T^2sp2} into \eqref{aoisp3}, the exact closed-form expression of the average AoI for the SP-AoR protocol can be presented in the following theorem.
\begin{theorem}\label{Theorem 1}
	The average AoI of the SP-AoR protocol is given by
	\begin{equation}
	\bar{\Delta}_{SP}\!=\!\frac{\big[1\!-\!(1\!-\!p)(1\!-\!P_{3})\big]\!\cdot\!\big[1\!-\!(1\!-\!p)(1\!-\!P_{1})(1\!-\!P_{2})\big]}{p\!\cdot\!\big[pP_{1}\!+\!(1\!-\!p)P_{3}\!-\!(1\!-\!p)(1\!-\!P_{1})(1\!-\!P_{2})P_{3}\big]}.
	\end{equation}
\end{theorem}

\subsection{Optimization of Average AoI}
Since the channel gain of the relay links is generally better than that of the direct link in cooperative communication systems, we assume that $0<P_{1}<P_{2}<1$ and $0<P_{1}<P_{3}<1$ in the optimization. We then acquire the optimal generation probability $p$ that minimizes the average AoI given by the following theorem.
\begin{theorem} \label{Theorem 2}
	{The optimal generation probability that minimizes the average AoI of the SP-AoR protocol is given by
		\begin{equation}\label{psp}
		p^*_{SP}\!=\!\left\{
		\begin{aligned}
		&\frac{-\!\lambda\!+\!\sqrt{\lambda^2\!-\!4\mu \xi}}{2\mu},\\
		&\quad  \mathrm{if} \ 
		0\!<\!P_{1}\!<\!\frac{P_{2}\!+\!P_{2}P_{3}\!-\!\sqrt{(P_{2}\!-\!P_{2}P_{3})^2\!+\!4P_{2}P_{3}}}{2(P_{2}\!-\!1)},\\
		&1, \\
		&\quad  \mathrm{if} \ 
		\frac{P_{2}\!+\!P_{2}P_{3}\!-\!\sqrt{(P_{2}\!-\!P_{2}P_{3})^2\!+\!4P_{2}P_{3}}}{2(P_{2}\!-\!1)}\!<\!P_{1}\!<\!1,
		\end{aligned} \right.
		\end{equation} 
		\noindent where $\lambda\!=\!-2P_{3}(P_{1}\!+\!P_{2}\!-\!P_{1}P_{2})(P_{1}\!-\!P_{1}P_{3}\!-\!P_{2}P_{3}\!+\!P_{1}P_{2}P_{3})$, $\mu\!=\!P_{2}P_{3}(1\!-\!P_{2}P_{3})\!-\!P_{1}^2(1\!-\!P_{2})(1\!-\!P_{3})^2\!-\!P_{1}P_{2}P_{3}^2(1\!-\!P_{2})(2\!-\!P_{1})\!-\!P_{1}P_{2}(1\!-\!P_{3})$ and $\lambda^2\!-\!4\mu \xi\!=\!4P_{2}P_{3}^2(P_{3}\!-\!P_{1})(1\!-\!P_{1})(P_{1}\!+\!P_{2}\!-\!P_{1}P_{2})^2$.}
\end{theorem}
\begin{proof} 
	See Appendix \ref{prof3}.
\end{proof} 

\begin{remark} \label{remark1}
	{For the second case in Theorem \ref{Theorem 2}, the average AoI is minimized by setting $p=1$ and it is given by
		\begin{equation}\label{SP_AoI}
		\bar\Delta_{SP}^{min}=\frac{1}{P_{1}}.
		\end{equation}
		Note that in the considered system, the minimum possible time required for a successful status update transmission through the $S$-$D$ link is one time slot, while that through the $S$-$R$-$D$ link is at least two time slots. Therefore, when the transmission success probability of the $S$-$D$ link (i.e., $P_{1}$) is relatively large, having more status updates successfully delivered through the $S$-$D$ link is beneficial to reduce the average AoI. If the generation probability of the status updates is 1, S always has a new status update to transmit (i.e., the generate-at-will model \cite{b9}). This means that in the SP-AoR protocol, the update received by $R$ in the previous time slot is always preempted. As a result, all the status updates received by $D$ can only come from the $S$-$D$ link, which minimizes the average AoI.
		
		On the other hand, for the first case in Theorem \ref{Theorem 2}, the transmission success probability of the $S$-$D$ link is relatively small compared to the second case. Hence, the updates successfully delivered through the $S$-$R$-$D$ link also have a significant contribution on minimizing the average AoI. As mentioned above, if the generation probability is 1, the updates can only be successfully delivered to $D$ through the $S$-$D$ link. Obviously, this is no longer optimal for minimizing the average AoI in this case, which means that $p=1$ should not be the optimal status generation probability to minimize the average AoI.}
\end{remark}

\vspace{-0.3cm}
\section{Relay-Prioritized AoR Protocol}
In this section, a low-complexity relay-prioritized AoR (RP-AoR) protocol is proposed to reduce the AoI of the considered system. Based on the analysis of the AoI evolution, we derive the closed-form expression of the average AoI for the RP-AoR protocol. The optimal status generation probability $p$ is also analyzed to minimize the average AoI.

\vspace{-0.2cm}
\subsection{Protocol Description}
Based on Theorem \ref{Theorem 2} and Remark \ref{remark1}, if the considered system suffers from severe channel fading on the $S$-$D$ link while the generation of new status updates at $S$ is frequent, the SP-AoR protocol may not have good performance in reducing the AoI. Inspired by this case, we propose the RP-AoR protocol, in which the retransmission from $R$ will not be preempted by a new status update arrival at $S$. Different from the SP-AoR protocol, if $R$ has a status update to forward, $S$ coordinates the system to keep transmitting with $\mathbf{O}_{\mathbf{R}}$ until it receives an ACK from $D$ (i.e., the forwarded update is successfully received at $D$). After the reception of the ACK, if $S$ has a fresher status update than that at $D$, $S$ coordinates the system to transmit with $\mathbf{O}_{\mathbf{S}}$. Otherwise, the considered system is coordinated by $S$ to choose $\mathbf{O}_{\mathbf{N}}$.

\vspace{-0.3cm}
\subsection{Analysis of Average AoI}
For the ease of understanding the AoI evolution of the RP-AoR protocol, we illustrate an example staircase path for 14 consecutive time slots with an initial value of one in Fig. 2. Note that $t_k$, $t_k'$, $t_*$, $S_k$, $T_k$, $Z_k$ and $W_k$ have the same definitions as those in Fig. 1. To facilitate the calculation of the AoI for the RP-AoR protocol, we further denote by $H_k$ the time that the $k$th status update received by $D$ waits in the system before being served, and denote by $Y_k$ the total time that the $k$th update received by $D$ spends in the system. For convenience, we refer to the system where the local AoI at neither $S$ nor $R$ is lower than that at $D$ (i.e., neither $S$ nor $R$ has a fresher status update than that at $D$ to transmit) as an empty system in the rest of this section.
	


\begin{figure}
	\centering
	\includegraphics[width=0.9\linewidth]{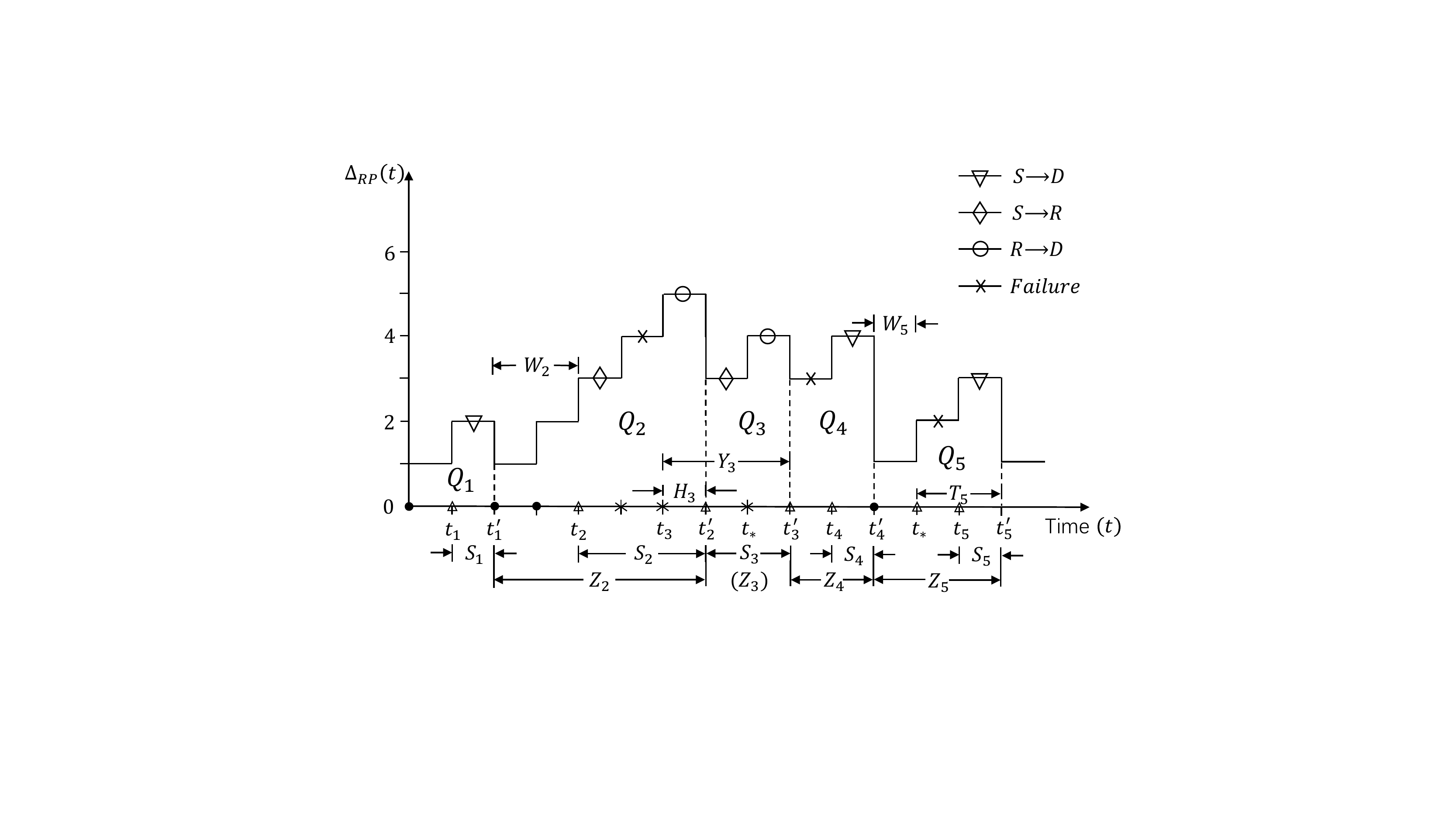}
	\caption{Sample staircase path of the AoI evolution for the RP-AoR protocol. We use $S\!\rightarrow\! D$, $S\!\rightarrow\! R$, $R\!\rightarrow\! D$, and Failure to denote the successful transmission through the $S$-$D$ link, the successful transmission through the $S$-$R$ link, the successful transmission through the $R$-$D$ link, and the failed transmission, respectively. Moreover, we use $\bullet$, {\scriptsize{$\triangle$}}, and $*$ to denote the events that none of the links are active, the $S$-$R$ link and the $S$-$D$ link are active, and the $R$-$D$ link is active, respectively.}
	\label{fig:aoirp}
	\vspace{-0.4cm}
\end{figure}

Similar to the process of obtaining \eqref{aoisp2}, we can express the average AoI of the RP-AoR protocol as 
\begin{equation}\label{aoirp1}
\bar{\Delta}_{RP}=\frac{\mathbb{E}[Y_{k-1}Z_{k}]}{\mathbb{E}[Z_{k}]}+\frac{\mathbb{E}[Z_{k}^2]}{2\mathbb{E}[Z_{k}]}-\frac{1}{2}.
\end{equation}
Note that in the RP-AoR protocol, if there is a status update waiting for transmission (i.e., the system is not empty) at $t_{k-1}'$, the $(k-1)$th status update received by $D$ can only be transmitted through the $S$-$R$-$D$ link.
However, if the system is empty at $t_{k-1}'$, the $(k-1)$th status update received by $D$ can be transmitted through either the $S$-$D$ link or the $S$-$R$-$D$ link. Therefore, for the two system states at $t_{k-1}'$, the distributions of the corresponding service time $S_{k-1}$ are different. In terms of the interdeparture time $Z_k$, it is affected by whether the $(k-1)$th status update received by $D$ leaves behind an empty system. As shown in Fig. 2, the system is not empty at $t_2'$, but empty at $t_4'$. During the interdeparture time $Z_3$, the system directly starts the service with the update left in the system at $t_2'$, while the system needs to wait for a status update generation to start the transmission within $Z_5$. This indicates that the system state at $t_{k-1}'$ determines whether the interdeparture time $Z_k$ includes the system waiting time $W_k$, and thus affects the distribution of $Z_k$. Based on the above analysis, it is apparent that the service time $S_{k-1}$ and the interdeparture time $Z_k$ are not independent in the RP-AoR protocol. As $Y_{k-1}=H_{k-1}+S_{k-1}$, we can find that $Y_{k-1}$ and $Z_k$ are not independent as well. To obtain the average AoI of the RP-AoR protocol, we derive the terms $\mathbb{E}[Y_{k-1}Z_{k}]$, $\mathbb{E}[Z_k]$ and $\mathbb{E}[Z_k^2]$ in the following. 

%
\subsubsection{Stationary Distribution of System State}
Prior to deriving the terms in the average AoI expression, we first need to find the stationary distribution of the system state.  We define $v(t) \triangleq \big(s(t), r(t)\big)$ as the system state of the considered cooperative system at the end of time slot $t$, where $s(t)$ indicates whether the local AoI at $S$ is lower than that at both $R$ and $D$ (i.e., whether $S$ has the freshest status update in the system to broadcast), and $r(t)$ indicates whether the local AoI at $R$ is lower than that at $D$ (i.e., whether $R$ has a fresher status update than that at $D$ to forward). Specifically, $s(t)=0$ and $s(t)=1$ indicate that at the end of time slot $t$, the status update at $S$ is not the freshest in the system and $S$ has the freshest status update in the system, respectively. Moreover, $r(t)=0$ and $r(t)=1$ indicate that at the end time slot $t$, $R$ has no fresher status updates than that at $D$ to forward, and $R$ has a fresher status update than that at $D$ to forward, respectively. Fig. \ref{Markov} depicts the state transition diagram. Based on the state transition, we can find that $\{v(t), t=1,2,...\}$ is an irreducible, positive-recurrent embedded Markov chain.


\begin{figure}
	\centering
	\includegraphics[width=0.9\linewidth]{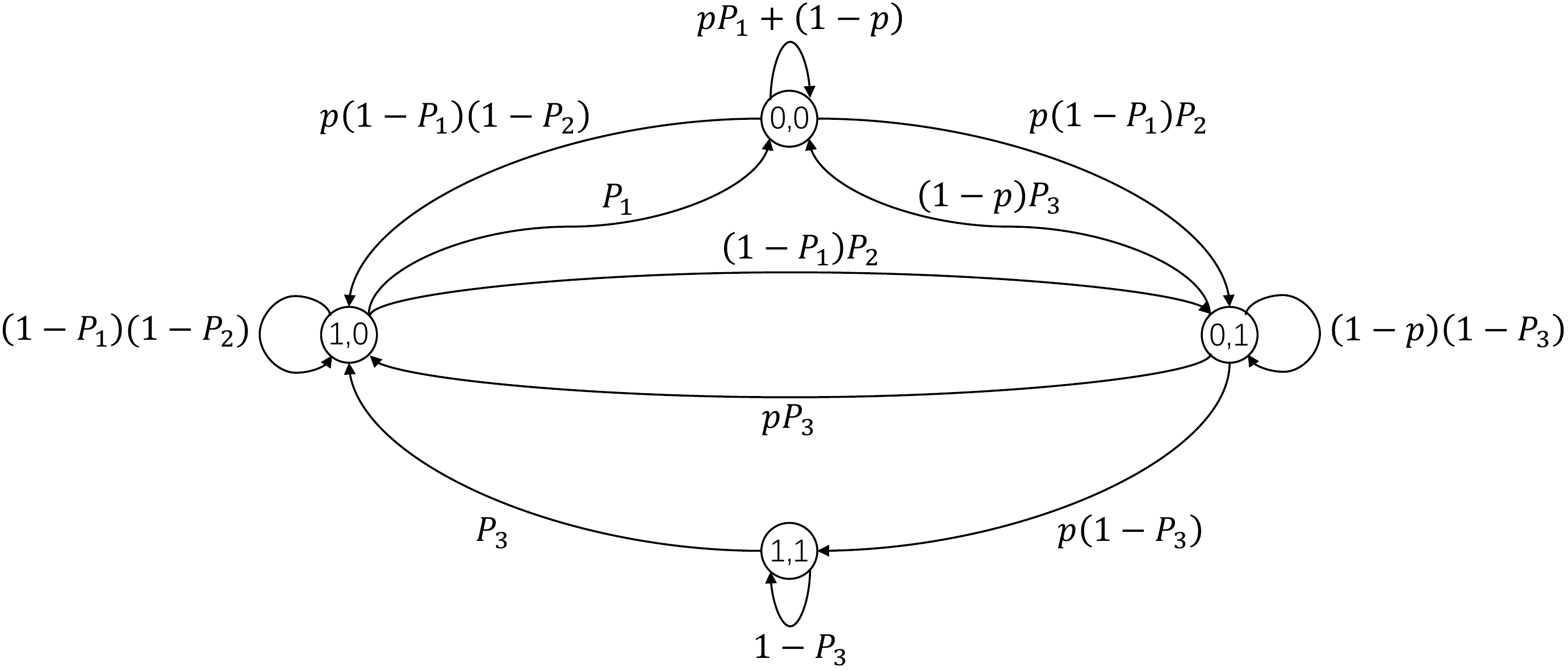}
	\caption{Markov chain for the system state $v(t)$ of the considered cooperative IoT system.}
	\label{Markov}
	\vspace{-0.4cm}
\end{figure}

%
%
%

With the transition probabilities in Fig. \ref{Markov}, we derive the stationary distribution of the Markov chain for the RP-AoR protocol, given by
\begin{subequations}
		\begin{equation}
		\pi_{0,0}=\frac{P_3\cdot\left[pP_1+(1-p)P_3\right]}{P_3(1-\alpha)(1+\beta)+p^2P_2(1-P_1)},
		\end{equation}
		\begin{equation}
		\pi_{0,1}=\frac{pP_2P_3(1-P_1)}{P_3(1-\alpha)(1+\beta)+p^2P_2(1-P_1)},
		\end{equation}
		\begin{equation}
		\pi_{1,0}=\frac{p^2P_3(1-P_1)+pP_3^2\beta}{P_3(1-\alpha)(1+\beta)+p^2P_2(1-P_1)},
		\end{equation}
		\begin{equation}
		\pi_{1,1}=\frac{p^2P_2(1-P_1)(1-P_3)}{P_3(1-\alpha)(1+\beta)+p^2P_2(1-P_1)},
		\end{equation}
\end{subequations}
where $\alpha$ and $\beta$ are defined in Section \ref{3B}.

\subsubsection{First Moment of the Interdeparture Time $\mathbb{E}[Z_k]$}
To evaluate $\mathbb{E}[Z_{k}]$, we condition on whether the $(k-1)$th status update received by $D$ leaves behind an empty system. We denote this event as $\Upsilon_{k-1}$ and its complement as  $\bar{\Upsilon}_{k-1}$. Based on the Bayes formula, the probability of $\Upsilon_{k-1}$ can be expressed as 
\begin{equation}\label{Pempt}
\mathrm{Pr}(\Upsilon_{k-1})=\frac{\mathrm{Pr}(A,B)}{\mathrm{Pr}(A)},
\end{equation}
where event $A$ denotes that $D$ receives a status update in time slot $t$, and event $B$ denotes that the reception at $D$ leaves behind an empty system.

In terms of the denominator, event $A$ occurs in four cases, which have different system states $v'$ at the end of the previous time slot (i.e., time slot $t-1$) and different corresponding conditions during the current time slot (i.e., time slot $t$): $(a)$ When $v'=(0,0)$, a new status update is generated at $S$ and successfully received by $D$ through the $S$-$D$ link during the current time slot; $(b)$ When $v'=(0,1)$, the update stored at $R$ is successfully forwarded to $D$ during the current time slot; $(c)$ When $v'=(1,0)$, a status update at $S$ is transmitted to $D$ through the $S$-$D$ link during the current time slot; $(d)$ When $v'=(1,1)$, the update stored at $R$ is successfully forwarded to $D$ during the current time slot.  In terms of the numerator, when event $B$ occurs together with event $A$, some of the above cases no longer meet the requirements. Specifically, if the successful reception at $D$ is through the $R$-$D$ link, the case where $S$ has the freshest status update in the system at the time of the reception are no longer applicable. Therefore, it is obvious that case $(d)$ will not appear in the calculation of the numerator. In addition, in case $(b)$, there should be no new status updates generated at $S$ during the forward by $R$.
Based on the above analysis, we can rewrite \eqref{Pempt} as 
\begin{equation}\label{P_empty}
\begin{aligned}
\mathrm{Pr}(\Upsilon_{k-1})&=\frac{\pi_{0,0}\cdot pP_1+\pi_{0,1}\cdot (1-p)P_3+\pi_{1,0}\cdot P_1}{\pi_{0,0}\cdot pP_1+\pi_{0,1}\cdot P_3+\pi_{1,0}\cdot P_1+\pi_{1,1}\cdot P_3}\\
&=\frac{pP_1+P_3(1-p-\beta)}{(1-\alpha)\cdot[1-(1-P1)(1-P2)]},
\end{aligned}
\end{equation}
and $\mathrm{Pr}(\bar{\Upsilon}_{k-1})=1-\mathrm{Pr}(\Upsilon_{k-1})$.


\begin{figure*}
	\begin{equation}\label{T}
	\begin{aligned}
	\mathbb{E}[T_k]&=\sum_{l=1}^{\infty}(1-p)^{l-1}(1-P_1)^{l-1}(1-P_2)^{l-1}P_1\cdot l\\
	&\quad + \sum_{n=1}^{\infty}\sum_{m=1}^{\infty}(1-p)^{m-1}(1-P1)^m(1-P_2)^{m-1}P_2(1-P_3)^{n-1}P_3\cdot (m+n)\\
	&\quad +\sum_{l=1}^{\infty}(1-P_1)^l(1-P_2)^l(1-p)^{l-1}p\cdot \big(l+\mathbb{E}\big[T_k'\big]\big).
	\end{aligned}
	\end{equation}
	\vspace{-0.3cm}	
	\noindent\rule[0\baselineskip]{\textwidth}{0.4pt}
\end{figure*}

Under the condition of $\Upsilon_{k-1}$, the system is empty at $t_{k-1}'$. Therefore, the corresponding interdeparture time $Z_k$ is the sum of the system waiting time for the generation of the first status update after $t_{k-1}'$ and the system transmission time from the generation to $t_k'$. We thus have $\mathbb{E}[Z_k|\Upsilon_{k-1}]=\mathbb{E}[W_k|\Upsilon_{k-1}]+\mathbb{E}[T_k|\Upsilon_{k-1}]$. Similarly, the first moment of the interdeparture time conditioned on $\bar{\Upsilon}_{k-1}$ can be expressed as $\mathbb{E}\big[Z_k|\bar{\Upsilon}_{k-1}\big]=\mathbb{E}\big[T_k|\bar{\Upsilon}_{k-1}\big]$. In terms of $\mathbb{E}[Z_k|\Upsilon_{k-1}]$, the waiting time $W_k$ only depends on the generation probability $p$ since the $(k-1)$th status update received by $D$ leaves behind an empty system. As the status updates are generated according to a Bernoulli process, $W_{k}$ conditioned on $\Upsilon_{k-1}$ follows a geometric distribution with parameter $p$, and its first moment can be readily given by $\mathbb{E}[W_k|\Upsilon_{k-1}]=(1-p)/p$. We then move on to the calculation of the terms $\mathbb{E}[T_k|\Upsilon_{k-1}]$ and $\mathbb{E}\big[T_k|\bar{\Upsilon}_{k-1}\big]$.


Under both conditions of $\Upsilon_{k-1}$ and $\bar{\Upsilon}_{k-1}$, $T_k$ behaves as the following three possible cases: $(a)$ The update is successfully received by $D$ through the $S$-$D$ link without being preempted; $(b)$ The update is successfully received by $D$ through the $S$-$R$-$D$ link without being preempted; $(c)$ The update is preempted by a new status update before being successfully received by either $R$ or $D$, and the new update may be preempted by multiple new updates. 
Therefore, the first moment of $T_{k}$ is the same under both conditions and we have $\mathbb{E}[T_k|\Upsilon_{k-1}]=\mathbb{E}\big[T_k|\bar{\Upsilon}_{k-1}\big]=\mathbb{E}[T_k]$. Based on the above possible cases, $\mathbb{E}[T_k]$ can be expressed by \eqref{T} at the top of the next page.

After some manipulations by applying the recursive method \cite{b8} and the finite sum equations given in \cite[Eqs. (0.112) and (0.113)]{b7}, we have
\begin{equation}\label{T_simplify}
\mathbb{E}[T_k]=\frac{P_3+P_2(1-P_1)}{P_3[1-(1-P_1)(1-P_2)]}.
\end{equation}
By combing $\mathbb{E}[W_k|\Upsilon_{k-1}]=(1-p)/p$ and \eqref{T_simplify}, we attain closed-form expressions for the first moment of $Z_k$, given by
\begin{subequations}\label{Z}
		\begin{equation}
\mathbb{E}\big[Z_{k}|
\bar{\Upsilon}_{k-1}\big]=\frac{P_2(1-P_1)+P_3}{P_3[1-(1-P_1)(1-P_2)]}, 
\end{equation}
and
\begin{equation}
\mathbb{E}[Z_{k}|\Upsilon_{k-1}]=\frac{1-p}{p}+\mathbb{E}\big[Z_{k}|
\bar{\Upsilon}_{k-1}\big].
\end{equation}
\end{subequations}

\subsubsection{Second Moment of the Interdeparture Time $\mathbb{E}\big[Z_{k}^{2}\big]$}
Recall that under the conditions of $\Upsilon_{k-1}$ and $\bar{\Upsilon}_{k-1}$, we have $Z_k=W_k+T_k$ and $Z_k=T_k$, respectively. Therefore, the corresponding second moment of $Z_k$ can be expressed as
\begin{subequations}\label{Z^2}
	\begin{equation}
	\resizebox{.98\hsize}{!}{$\begin{aligned}
	\mathbb{E}\big[Z_k^2|\Upsilon_{k-1}\big]&\!=\!\mathbb{E}\big[(W_k\!+\!T_k)^2|\Upsilon_{k-1}\big]\\
	&\!=\!\mathbb{E}\big[W_k^2|\Upsilon_{k-1}\big]\!+\!2\mathbb{E}[W_kT_k|\Upsilon_{k-1}]\!+\!\mathbb{E}\big[T_k^2|\Upsilon_{k-1}\big].
	\end{aligned}$}
	\end{equation}
	and
	\begin{equation}
	\mathbb{E}\big[Z_k^2|\bar{\Upsilon}_{k-1}\big]=\mathbb{E}\big[T_k^2|\bar{\Upsilon}_{k-1}\big].
	\end{equation}
\end{subequations}
It is readily to find that $W_k$ and $T_k$ are conditionally independent under $\Upsilon_{k-1}$. Thus, we have $\mathbb{E}[W_kT_k|\Upsilon_{k-1}]=\mathbb{E}[W_k|\Upsilon_{k-1}]\cdot \mathbb{E}[T_k|\Upsilon_{k-1}]$. As $W_k$ follows a geometric distribution with parameter $p$, we have $\mathbb{E}\big[W_k^2|\Upsilon_{k-1}\big]\!=\!(p^2-3p+2)/p^2$. In order to evaluate $\mathbb{E}\big[Z_k^2\big]$, the only remaining task is to calculate $\mathbb{E}\big[T_k^2|\Upsilon_{k-1}\big]$ and $\mathbb{E}\big[T_k^2|\bar{\Upsilon}_{k-1}\big]$.

\begin{figure*}
	\begin{equation}\label{T^2}
	\begin{aligned}
	\mathbb{E}\big[T_k^2\big]&=\sum_{l=1}^{\infty}(1-p)^{l-1}(1-P_1)^{l-1}(1-P_2)^{l-1}P_1\cdot l^2\\
	&\quad + \sum_{n=1}^{\infty}\sum_{m=1}^{\infty}(1-p)^{m-1}(1-P1)^m(1-P_2)^{m-1}P_2(1-P_3)^{n-1}P_3\cdot (m^2+2mn+n^2)\\
	&\quad +\sum_{l=1}^{\infty}(1-P_1)^l(1-P_2)^l(1-p)^{l-1}p\cdot \left(l^2+2l\cdot \mathbb{E}[T_k]+\mathbb{E}\big[(T_k')^2\big]\right).
	\end{aligned}
	\end{equation}	
	\vspace{-0.4cm}
	\noindent\rule[0\baselineskip]{\textwidth}{0.4pt}
\end{figure*}

Similar to \eqref{T}, $\mathbb{E}\big[T_k^2|
\Upsilon_{k-1}\big]$ and $\mathbb{E}\big[T_k^2|\bar{\Upsilon}_{k-1}\big]$ have the same value, which can be expressed by \eqref{T^2} at the top of the next page. Similar to the process of obtaining \eqref{T_simplify} from \eqref{T}, we can rewrite $\mathbb{E}\big[T_k^2\big]$ as
\begin{equation}\label{T^2_simplify}
\resizebox{.98\hsize}{!}{$\begin{aligned}
	\mathbb{E}\big[T_k^2\big]\!=\!\ &\frac{1}{P_3^2\!\cdot\![1\!-\!(1\!-\!P_1)(1\!-\!P_2)]^2}\\
	&\times\!\Big\{\!P_2^2\!\cdot\!(1\!-\!P_1)^2\!\cdot\!(2\!-\!P_3)\!+\!P_3^2\!\cdot\![1\!+\!(1\!-\!P_1)(1\!-\!P_2)]\\
	&\quad \ +\!P_2(2\!-\!P_1)\!\cdot\![1\!-\!(1\!-\!P_1)(1\!-\!P_3)]\!-\!P_1^2P_2\Big\}.
	\end{aligned}$}
\end{equation}
Based on the above analysis, the second moment of $Z_k$ can be expressed as
\begin{subequations}\label{Z^2_final}
		\begin{equation}
		\resizebox{.98\hsize}{!}{$\begin{aligned}
			\mathbb{E}\big[Z_{k}^2|\bar{\Upsilon}_{k-1}\big]\!=\!\ &\frac{1}{P_3^2\!\cdot\![1\!-\!(1\!-\!P_1)(1\!-\!P_2)]^2}\\
			&\!\times\!\Big\{\!P_2^2\!\cdot\!(1\!-\!P_1)^2\!\cdot\!(2\!-\!P_3)\!+\!P_3^2\!\cdot\![1\!+\!(1\!-\!P_1)(1\!-\!P_2)]\\
			&\quad \ +\!P_2(2\!-\!P_1)\!\cdot\![1\!-\!(1\!-\!P_1)(1\!-\!P_3)]\!-\!P_1^2P_2\Big\},
			\end{aligned}$}
\end{equation}
	and
	\begin{equation}
		\begin{aligned}
			\mathbb{E}\big[Z_{k}^2|\Upsilon_{k-1}\big]&=\mathbb{E}\big[Z_{k}^2|\bar{\Upsilon}_{k-1}\big]+\frac{p^2-3p+2}{p^2}\\
			&\,\quad +\frac{(2-2p)\cdot[P_2(1-P_1)+P_3]}{p P_3\cdot[1-(1-P_1)(1-P_2)]}.
		\end{aligned}
	\end{equation}
\end{subequations}

\subsubsection{$\mathbb{E}[Y_{k-1}Z_k]$}
We finally evaluate the term $\mathbb{E}[Y_{k-1}Z_k]$. The total time that the $(k-1)$th status update received by $D$ spends in the system is the sum of its waiting time before being served and its service time. Therefore, we have  $Y_{k-1}=H_{k-1}+S_{k-1}$. Under the condition of $\Upsilon_{k-1}$, the first moment of $Y_{k-1}$ can be expressed by $\mathbb{E}[Y_{k-1}|\Upsilon_{k-1}]=\mathbb{E}[H_{k-1}|\Upsilon_{k-1}]+\mathbb{E}[S_{k-1}|\Upsilon_{k-1}]$. To further compute $\mathbb{E}[Y_{k-1}Z_k|\Upsilon_{k-1}]$, we have the following lemma.
\begin{lemma}\label{lemma2}
	In the RP-AoR protocol, if the $(k-1)$th status update received by $D$ leaves behind an empty system, the total time that the $(k-1)$th status update received by $D$ spends in the system is independent of the interdeparture time between the $(k-1)$th and the $k$th received status update at $D$ such that
	\begin{equation} \label{lemma2equation}
	\mathbb{E}[Y_{k-1}Z_k|\Upsilon_{k-1}]=\mathbb{E}[Y_{k-1}|\Upsilon_{k-1}]\cdot \mathbb{E}[Z_k|\Upsilon_{k-1}].
	\end{equation}
\end{lemma}

\begin{proof}
	Since we have $Y_{k-1}=H_{k-1}+S_{k-1}$, Lemma \ref{lemma2} holds if the interdeparture time $Z_k$ is independent of both $H_{k-1}$ and $S_{k-1}$ under the condition of $\Upsilon_{k-1}$. In the RP-AoR protocol, $H_{k-1}$ is conditioned on whether the system is empty when the generation of the $(k-1)$th status update received by $D$ occurs, but as we mentioned, $Z_k$ is conditioned on whether the $(k-1)$th update received by $D$ leaves behind an empty system. This indicates that $Z_k$ is independent of $H_{k-1}$. Similar to the proof of Lemma \ref{lemma1}, we can readily prove that $Z_k$ is conditionally independent of $S_{k-1}$ under $\Upsilon_{k-1}$. This completes the proof.
\end{proof}
\noindent Similarly, we can have $\mathbb{E}\big[Y_{k-1}Z_k|\bar{\Upsilon}_{k-1}\big]=\mathbb{E}\big[Y_{k-1}|\bar{\Upsilon}_{k-1}\big]\cdot \mathbb{E}\big[Z_k|\bar{\Upsilon}_{k-1}\big]$, where $\mathbb{E}\big[Y_{k-1}|\bar{\Upsilon}_{k-1}\big]=\mathbb{E}\big[H_{k-1}|\bar{\Upsilon}_{k-1}\big]+\mathbb{E}\big[S_{k-1}|\bar{\Upsilon}_{k-1}\big]$ and $\mathbb{E}\big[Z_k|\bar{\Upsilon}_{k-1}\big]=\mathbb{E}\big[T_k|\bar{\Upsilon}_{k-1}\big]$. The proof is similar to the above and is thus omitted for brevity. Since $\mathbb{E}[Z_k|\Upsilon_{k-1}]$ and $\mathbb{E}\big[Z_k|\bar{\Upsilon}_{k-1}\big]$ is derived in \eqref{Z}, the remaining tasks of calculating $\mathbb{E}[Y_{k-1}Z_{k}]$ are to evaluate $\mathbb{E}[Y_{k-1}|\Upsilon_{k-1}]$ and $\mathbb{E}\big[Y_{k-1}|\bar{\Upsilon}_{k-1}\big]$.

Recall that we have $Y_{k-1}=H_{k-1}+S_{k-1}$. To compute $\mathbb{E}[Y_{k-1}]$, we first consider the packet waiting time before being served, i.e., $H_{k-1}$. As we mentioned, in the RP-AoR protocol, $H_{k-1}$ is conditioned on whether the system is empty when the generation of the $(k-1)$th status update received by $D$ occurs.
In other words, it is conditioned on the system state when the $(k-2)$th update received by $D$ leaves the system, rather than the system state when the $(k-1)$th update received by $D$ leaves the system, which indicates that $\mathbb{E}[H_{k-1}|\Upsilon_{k-1}]=\mathbb{E}\big[H_{k-1}|\bar{\Upsilon}_{k-1}\big]=\mathbb{E}[H_{k-1}]$. Obviously, if the $(k-2)$th status update received by $D$ leaves the system empty, we have $H_{k-1}=0$. Otherwise, $H_{k-1}$ is determined by whether the update left in the system at $t_{k-2}'$ is successfully received by $D$. We denote by $\Theta$ the event that the left update is preempted during the transmission (i.e., not successfully received by $D$) and by $\bar{\Theta}$ its complement. The probability that the left update is preempted before the $j$th transmission starts can be expressed as
\vspace{-0.01cm}
\begin{equation}
P_{preempt}=(1-p)^{j-1}(1-P_1)^{j-1}(1-P_2)^{j-1}p.
\vspace{-0.01cm}
\end{equation}
Therefore, the probability of event $\Theta$ are given by
\vspace{-0.01cm}
\begin{equation}
\mathrm{Pr}(\Theta)=\sum_{j=1}^{\infty}P_{preempt}=\frac{p}{1-\beta},
\vspace{-0.01cm}
\end{equation}
and $\mathrm{Pr}(\bar{\Theta})=1-\mathrm{Pr}(\Theta)$. Under the condition of $\Theta$, a fresher status update preempts the update left in the system and starts to be served without waiting. Thus, we can have $H_{k-1}=0$ conditioned on $\Theta$, and the only remaining task to evaluate $\mathbb{E}[H_{k-1}]$ is to calculate $\mathbb{E}\big[H_{k-1}|\bar{\Theta}\big]$. If an update is left in the system at $t_{k-2}'$, it can only be generated at $S$ after the $(k-2)$th status update received by $D$ has arrived at $R$. This means that $\mathbb{E}\big[H_{k-1}|\bar{\Theta}\big]$ is the residual transmission time of the $(k-2)$th update received by $D$ on the $R$-$D$ link. After the update left in the system at $t_{k-2}'$ is generated at $S$, the probability that the transmission on the $R$-$D$ link performs $n$ times can be expressed as
\vspace{-0.02cm}
\begin{equation}
P_{n}=(1-p)^{n-1}(1-P_3)^{n-1}P_3.
\vspace{-0.01cm}
\end{equation}
The corresponding first moment of the transmission time, i.e., $\mathbb{E}\big[H_{k-1}|\bar{\Theta}\big]$, can be calculated by
\vspace{-0.01cm}
\begin{equation}
\mathbb{E}\big[H_{k-1}|\bar{\Theta}\big]=\frac{\sum_{n=1}^{\infty}P_n\cdot n}{\sum_{n=1}^{\infty}P_n}=\frac{1}{1-\alpha}.
\end{equation}
By conditioning on whether the $(k-2)$th status update received by $D$ leaves the system empty or not, we now evaluate the first moment of $H_{k-1}$ for the RP-AoR protocol, given by
\begin{equation}\label{H}
	\begin{aligned}
		\mathbb{E}[H_{k-1}]\!=&\Big\{\mathbb{E}\big[H_{k-1}|\bar{\Theta}\big]\cdot \mathrm{Pr}\big(\bar{\Theta}\big)+0 \times \mathrm{Pr}(\Theta)\Big\}\times \mathrm{Pr}\big(\bar{\Upsilon}_{k-1}\big)\\
		&\ +0 \times \mathrm{Pr}(\Upsilon_{k-1})\\
		&\!\!\!\!\!\!=\frac{p P_2(1-p)(1-P_1)}{(1-\alpha)^2(1-\beta)}.
	\end{aligned}
\end{equation}
%
%
%
%
After the calculation of $\mathbb{E}[H_{k-1}]$, we then move on to the other component of $Y_{k-1}$, i.e., $S_{k-1}$.


Under the condition of $\Upsilon_{k-1}$, the $(k-1)$th status update received by $D$ leaves behind an empty system. Therefore, the first moment of the service time $\mathbb{E}[S_{k-1}|\Upsilon_{k-1}]$ is composed of the average service time for the same two types of status updates as in the SP-AoR protocol. Based on \eqref{PL}$-$\eqref{Ssp2}, the first moment of the service time conditioned on $\Upsilon_{k-1}$ are given by
\begin{equation}\label{S1}
\mathbb{E}[S_{k-1}|\Upsilon_{k-1}]=\frac{1}{1-\beta}+\frac{1}{1-\alpha}\cdot\frac{\gamma}{P_{1}(1-\alpha)+\gamma}.
\end{equation}
Under the condition of $\bar{\Upsilon}_{k-1}$, the system is not empty when the $(k-1)$th status update received by $D$ leaves the system. This means that the status update must be transmitted through the $S$-$R$-$D$ link, and there must be fresher status updates generated at $S$ after it arrives at $R$. The first moment of the service time conditioned on $\bar{\Upsilon}_{k-1}$ can thus be expressed as
\begin{equation}\label{RS1}
	\resizebox{.98\hsize}{!}{$\begin{aligned}
\mathbb{E}\big[S_{k-1}|
\bar{\Upsilon}_{k-1}\big]\!=\!\frac{\sum_{n=1}^{\infty}\!\sum_{m=1}^{\infty}\!P\!\left(S\!\rightarrow\! R\!\rightarrow\! D|\bar{\Upsilon}_{k-1}\right)\!\cdot\! (m\!+\!n)}{\sum_{n=1}^{\infty}\!\sum_{m=1}^{\infty}\!P\!\left(S\!\rightarrow\! R\!\rightarrow\! D|\bar{\Upsilon}_{k-1}\right)},
\end{aligned}$}
\end{equation}
where $m$ and $n$ are defined in Section \ref{3B}. $P\left(S\rightarrow R\rightarrow D|\bar{\Upsilon}_{k-1}\right)$ is the probability that the status update is received by $D$ after transmitting $(m+n)$ times conditioned on $\bar{\Upsilon}_{k-1}$, which can be expressed as
\begin{equation}
\begin{aligned}
P\!\left(S\!\rightarrow \!R\!\rightarrow\! D|\bar{\Upsilon}_{k-1}\right)\!=\!\ &(1\!-\!p)^{m-1}(1\!-\!P_1)^m(1\!-\!P_2)^{m-1}P_2\\
&\!\times(1\!-\!P_3)^{n-1}P_3\left[1\!-\!(1\!-\!p)^n\right].
\end{aligned}
\end{equation}
By applying \cite[Eqs. (0.112) and (0.113)]{b7}, we can simplify \eqref{RS1} as
\begin{equation}\label{S2}
\mathbb{E}\big[S_{k-1}|\bar{\Upsilon}_{k-1}\big]=\frac{2}{P_3}+\frac{1}{1-\beta}-\frac{P_3^2(1-p)+p}{P_3(1-\alpha)}.
\end{equation}
Based on \eqref{H}, \eqref{S1}, and \eqref{S2}, we attain closed-form expressions for the first moment of $Y_{k-1}$, given by
\begin{subequations} \label{Y}
			\begin{equation}
				\begin{aligned}
					\mathbb{E}\big[Y_{k-1}|\bar{\Upsilon}_{k-1}\big]=\ &\frac{pP_2(1-p)(1-P_1)}{(1-\alpha)^2(1-\beta)}+\frac{1}{1-\beta}\\
					&\!+\frac{2}{P_3}-\frac{P_3^2(1-p)+p}{P_3(1-\alpha)},
				\end{aligned}
	\end{equation}
	and
			\begin{equation}
			\begin{aligned}
			\mathbb{E}[Y_{k-1}|\Upsilon_{k-1}]=\ &\frac{pP_2(1-p)(1-P_1)}{(1-\alpha)^2(1-\beta)}+\frac{1}{1-\beta}\\
			&+\frac{\gamma}{P_1(1-\alpha)^2+\gamma(1-\alpha)}.
			\end{aligned}
	\end{equation}
\end{subequations}

\subsubsection{Average AoI}

Based on the above analysis, the exact closed-form expression of the average AoI for the RP-AoR protocol can be presented in the following theorem.

\begin{theorem}\label{Theorem 3}
	The average AoI of the RP-AoR protocol is given by
\begin{equation}
	\resizebox{.98\hsize}{!}{$\begin{aligned}
			\bar{\Delta}_{RP}&\!=\!\frac{1}{\mathbb{E}[Z_k|\Upsilon_{k-1}]\!\cdot\! \mathrm{Pr}(\Upsilon_{k-1})\!+\!\mathbb{E}\big[Z_k|\bar{\Upsilon}_{k-1}\big]\!\cdot\! [1\!-\!\mathrm{Pr}(\Upsilon_{k-1})]}\\
			&\quad\!\times\Big\{\mathbb{E}[Y_{k-1}|\Upsilon_{k-1}]\cdot\mathbb{E}[Z_k|\Upsilon_{k-1}]\cdot \mathrm{Pr}(\Upsilon_{k-1})\\
			&\quad\quad \quad \!\!\! +\mathbb{E}\big[Y_{k-1}|\bar{\Upsilon}_{k-1}\big]\cdot\mathbb{E}\big[Z_k|\bar{\Upsilon}_{k-1}\big]\cdot \big[1\!-\!\mathrm{Pr}(\Upsilon_{k-1})\big]\\\
			&\quad \quad \quad \!\!\! +\frac{1}{2}\cdot \mathbb{E}\big[Z_k^2|\Upsilon_{k-1}\big]\cdot\mathrm{Pr}(\Upsilon_{k-1})\\
			&\quad \quad \quad \!\!\!+\frac{1}{2}\cdot \mathbb{E}\big[Z_k^2|\bar{\Upsilon}_{k-1}\big]\cdot\big[1-\mathrm{Pr}(\Upsilon_{k-1})\big]\Big\}-\frac{1}{2},
	\end{aligned}$}
\end{equation}
		\noindent where $\mathrm{Pr}(\Upsilon_{k-1})$ is given in \eqref{P_empty}, $\mathbb{E}[Z_k|\Upsilon_{k-1}]$ and $\mathbb{E}\big[Z_k|\bar{\Upsilon}_{k-1}\big]$ are given in \eqref{Z}, $\mathbb{E}\big[Z_k^2|\Upsilon_{k-1}\big]$ and $\mathbb{E}\big[Z_k^2|\bar{\Upsilon}_{k-1}\big]$ are given in \eqref{Z^2_final}, and $\mathbb{E}[Y_{k-1}|\Upsilon_{k-1}]$ and $\mathbb{E}\big[Y_{k-1}|\bar{\Upsilon}_{k-1}\big]$ are given in \eqref{Y}.

\end{theorem}

\subsection{Optimization of Average AoI}
The optimal generation probability $p$ that minimizes the average AoI of the RP-AoR protocol is given by the following theorem.
\begin{theorem}\label{Theorem 4}
In the RP-AoR protocol, the optimal generation probability that minimizes the average AoI is given by 
\begin{equation}
	p_{RP}^*=1.
\end{equation}	
\end{theorem}
\begin{proof}
	Recall that in the RP-AoR protocol, the system can have three transmission operations in each time slot, i.e.,  $\mathbf{O}_{\mathbf{S}}$,  $\mathbf{O}_{\mathbf{R}}$ and  $\mathbf{O}_{\mathbf{N}}$. If the system decides to transmit with $\mathbf{O}_{\mathbf{S}}$ in one time slot, generating new status update at $S$ in each slot (i.e., $p_{RP}^*=1$) ensures that the transmitted update is always the freshest. If the system chooses $\mathbf{O}_{\mathbf{R}}$ in one time slot, the forwarded update by $R$ will not be preempted by new status updates generated at $S$. In this case, $p_{RP}^*=1$ keeps the status update waiting for transmission at $S$ as fresh as possible. Furthermore, when the generation probability at $S$ is 1, $\mathbf{O}_{\mathbf{N}}$ will never be adopted by the system, which is apparently beneficial for the AoI minimization. This completes the proof.   
\end{proof}

\begin{remark}\label{remark_2}
		For the RP-AoR protocol, the average AoI is minimized by setting $p=1$ (i.e., the generate-at-will model) and it is given by 
		\begin{equation}\label{RP_AoI}
		\bar{\Delta}^{min}_{RP}\!=\!\frac{P_2(1\!-\!P_1)}{P_3(P_2\!+\!P_3\!-\!P_1P_2)}\!+\!\frac{P_2\!+\!P_3\!-\!P_1P_2}{P_3[1\!-\!(1\!-\!P_1)(1\!-\!P_2)]}.
		\end{equation}
		If the considered system adopts the generate-at-will model, by solving the simultaneous equations of \eqref{SP_AoI} and \eqref{RP_AoI}, we have: $(a)$ When $P_1<f(P_2,P_3)$, where $f(P_2,P_3)=\frac{2P_2+P_3+P_2P_3-\sqrt{P_2^2(P_3-2)^2+P_3(8P_2+5P_3-6P_2P_3)}}{4P_2-2}$, the RP-AoR protocol has a better AoI performance than the SP-AoR protocol; $(b)$ When $P_1>f(P_2,P_3)$, the SP-AoR protocol has a better AoI performance that the RP-AoR protocol.	
		
		
\end{remark}

\section{Numerical Results and Discussions}
In this section, we present simulation results to validate the theoretical analysis and the effectiveness of the two proposed AoR protocols in the considered cooperative IoT  system. Each simulation curve presented in the section is averaged over $10^7$ time slots.

We first plot the average AoI of the proposed AoR protocols and the optimal MDP policy against the generation probability $p$, as shown in Fig. \ref{simulation_1}, for different transmission success probabilities (i.e., $P1$, $P2$ and $P_3$). Specifically, we consider five combinations of $(P_1,P_2,P_3)$, which are $(0.2,0.3,0.3)$, $(0.2,0.3,0.8)$, $(0.2,0.8,0.3)$, $(0.2,0.8,0.8)$, and $(0.7,0.8,0.8)$, to represent the systems where all channels suffer from severe fading, both the $S$-$D$ link and the $S$-$R$ link suffer from severe fading while the $R$-$D$ link is in good condition, both the $S$-$D$ link and the $R$-$D$ link suffer from severe fading while the $S$-$R$ link is in good condition, the $S$-$D$ link suffers from severe fading while both the $S$-$R$ link and the $R$-$D$ link are in good condition, and all channels are in good condition, respectively. It can be observed that in all the cases, the simulation results of the average AoI for both the SP-AoR protocol and the RP-AoR protocol are consistent with the corresponding analytical results. This observation verifies the correctness of Theorems \ref{Theorem 1} and \ref{Theorem 3}. We can also observe that in Fig. \ref{simulation_1} (a), the optimal generation probability to minimize the average AoI of the SP-AoR protocol in the case of $(P_1,P_2,P_3)=(0.2,0.3,0.3)$, $(P_1,P_2,P_3)=(0.2,0.8,0.8)$, and $(P_1,P_2,P_3)=(0.7,0.8,0.8)$ are $p=1$, $p=0.616$, and $p=1$, respectively. In Fig. \ref{simulation_1} (b), the optimal generation probability of the SP-AoR protocol in the case of $(P_1,P_2,P_3)=(0.2,0.3,0.8)$ and $(P_1,P_2,P_3)=(0.2,0.8,0.3)$ are $p=0.662$ and $p=0.826$, respectively. The optimal generation probabilities coincide well with the results calculated by the formula in Theorem \ref{Theorem 2}. It is also worth mentioning that the average AoI of the RP-AoR protocol is minimized when $p=1$ in all the cases, which are consistent with Theorem \ref{Theorem 4}.

\begin{figure}
	\centering
	\subfigure[\scriptsize{$(P_1, P_2, P_3)=(0.2,0.3,0.3), (0.2,0.8,0.8)$ and $(0.7,0.8,0.8)$}]{\includegraphics[width=0.85\linewidth]{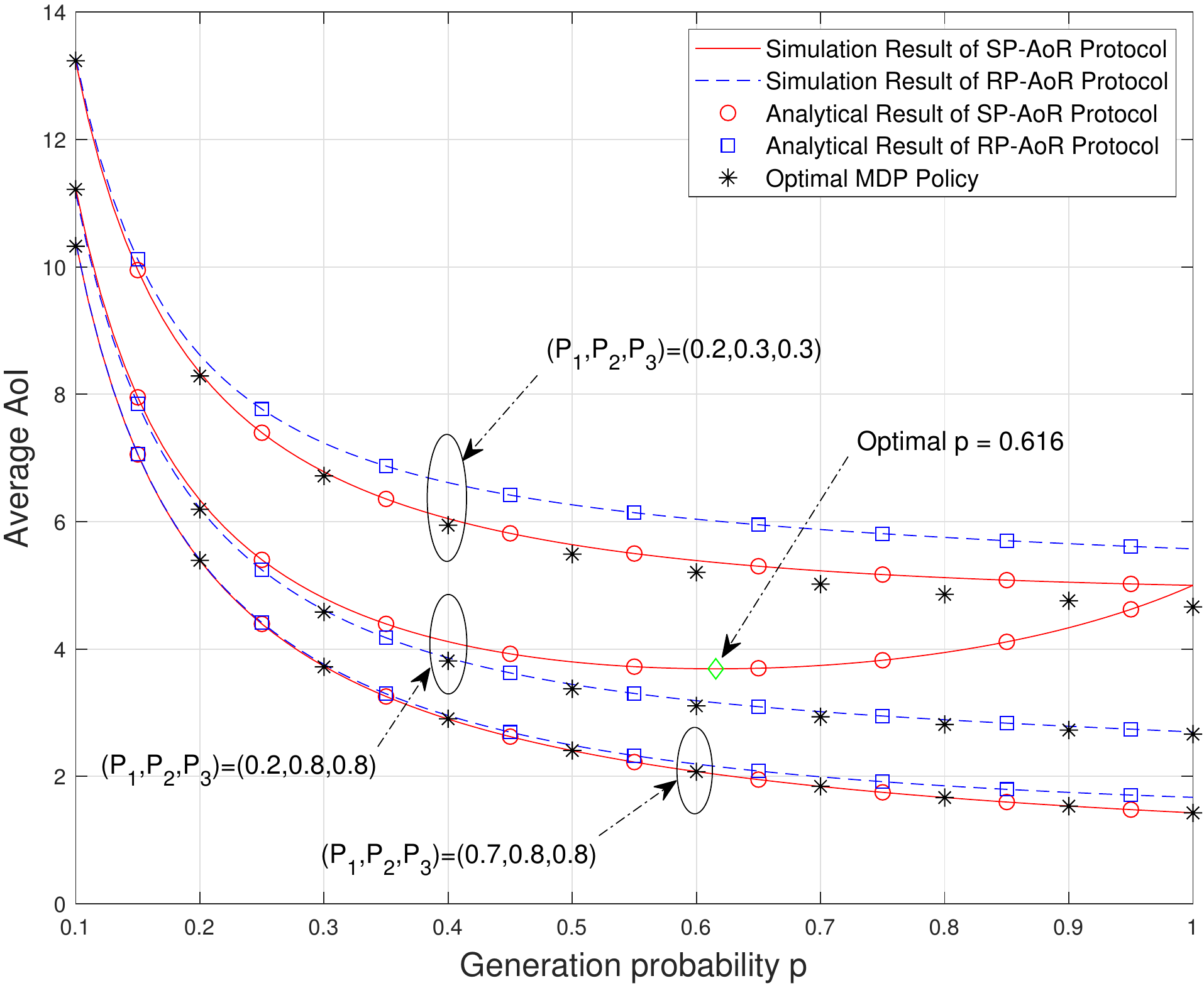}}
	\qquad\qquad
	\subfigure[\scriptsize{$(P_1, P_2, P_3)=(0.2,0.3,0.8)$ and $ (0.2,0.8,0.3)$}]{\includegraphics[width=0.85\linewidth]{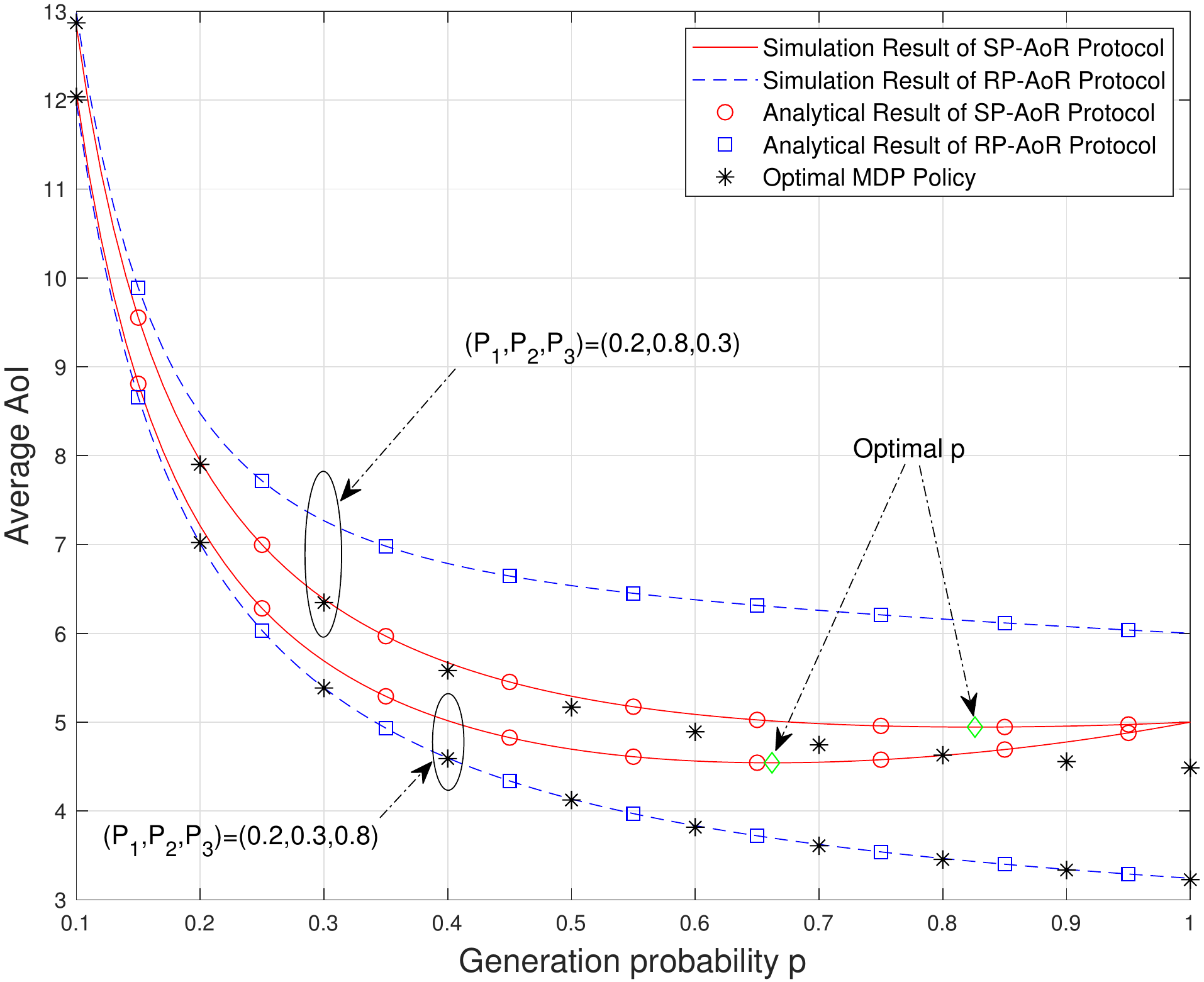}}
	\caption{Average AoI of the SP-AoR protocol, the RP-AoR protocol and the optimal MDP policy against the generation probability $p$ for different transmission success probabilities, i.e., $(P_1, P_2, P_3)$.}
	\label{simulation_1}
	\vspace{-0.5cm}
\end{figure}

Based on Fig. \ref{simulation_1}, we can see that the SP-AoR protocol and the RP-AoR protocol outperform each other in different cases. Note that in the considered cooperative IoT system, we can quickly determine which protocol to apply for better AoI performance with the help of the analytical results. More importantly, the protocol with better performance can achieve near-optimal performance compared with the optimal scheduling policy solved by the MDP tool. Specifically, the RP-AoR protocol has better AoI performance than the SP-AoR protocol when $(P_1, P_2,P_3)=(0.2,0.8,0.8)$ and $(P_1,P_2,P_3)=(0.2,0.3,0.8)$, as shown in Fig. \ref{simulation_1} (a) and Fig. \ref{simulation_1} (b), respectively. Also, the performance of the RP-AoR protocol is  almost the same as that of the optimal MDP policy in both cases. This is because, in both cases, the $S$-$D$ link suffers from severe channel fading while the $R$-$D$ link has good channel condition. This means that it is difficult to successfully transmit status updates to $D$ through the $S$-$D$ link, while the status updates forwarded by $R$ have a high probability of being received by $D$. In such cases, even if a new status update is generated at $S$, the system should give the priority to  $R$ for transmitting, which is consistent with the principle of the RP-AoR protocol. We can also observe that the SP-AoR protocol results in lower average AoI than the RP-AoR protocol in the case of $(P_1, P_2,P_3)=(0.2,0.3,0.3)$ and $(P_1, P_2,P_3)=(0.2,0.8,0.3)$, as shown in Fig. \ref{simulation_1} (a) and Fig. \ref{simulation_1} (b), respectively. The rationale behind this is that when neither $S$-$D$ link nor $R$-$D$ link has good channel condition, it is difficult for both $S$ and $R$ to successfully transmit status updates to $D$. Therefore, the SP-AoR protocol, in which $S$ has the priority to transmit newly generated status updates, is more beneficial for reducing the AoI at $D$ in such cases. In these two cases, although there is a performance gap between the SP-AoR protocol and the optimal MDP policy when the generation probability $p$ is large, the SP-AoR can still achieve near-optimal performance. In addition, as shown in Fig. \ref{simulation_1} (a), the SP-AoR protocol has a better AoI performance than the RP-AoR protocol when $(P_1,P_2,P_3)=(0.7,0.8,0.8)$, and the performance almost matches that of the optimal MDP policy. 


\begin{figure}
	\centering
	\includegraphics[width=0.85\linewidth]{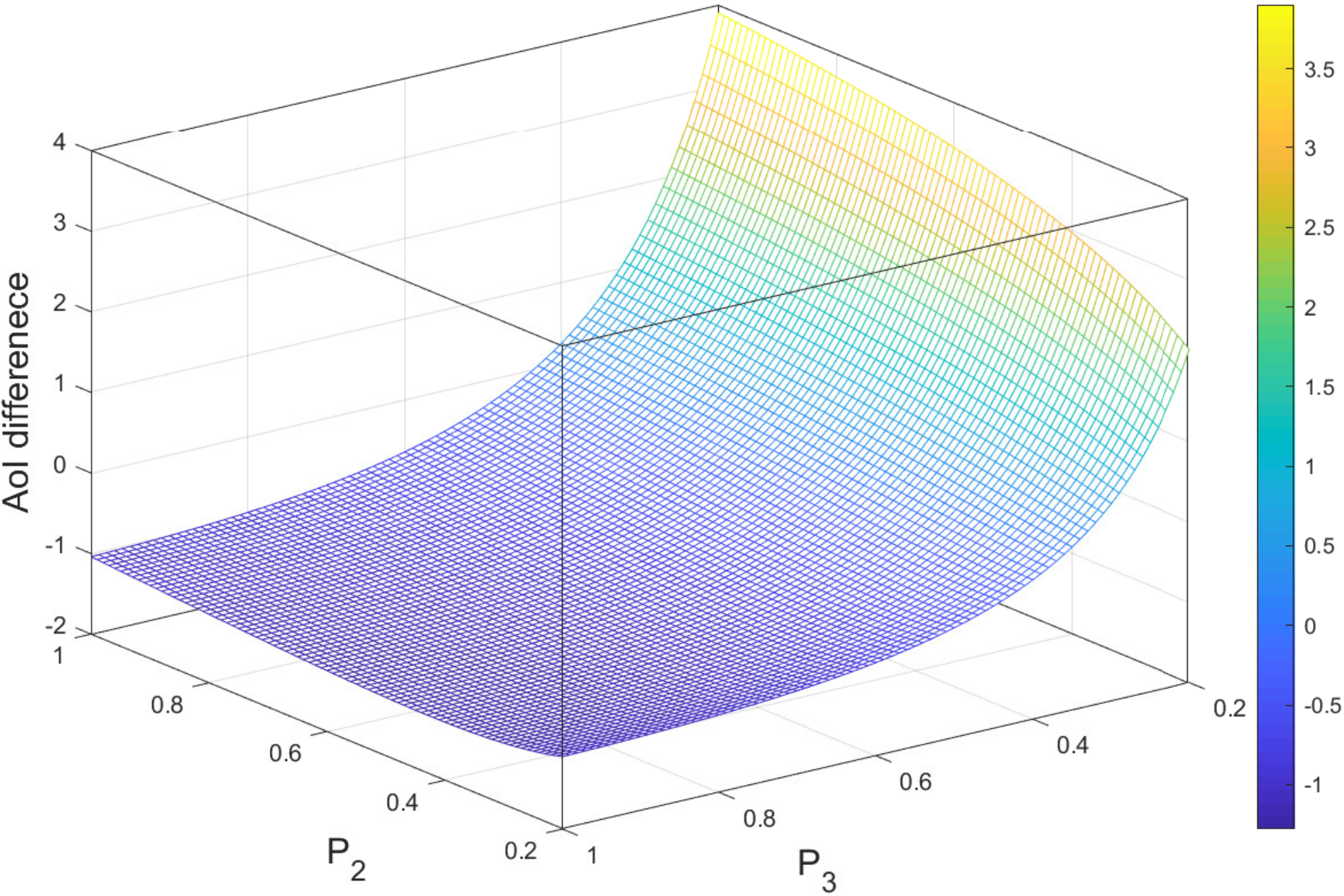}
	\caption{AoI difference between the RP-AoR protocol and the SP-AoR protocol versus the transmission success probabilities $P_2$ and $P_3$ when $P_1=0.2$ and $p=0.8$.}
	\label{simulation_2}
	\vspace{-1cm}
\end{figure}

Fig. \ref{simulation_2} shows the AoI difference between the RP-AoR protocol and the SP-AoR protocol against the transmission success probabilities of the $S$-$R$ link and the $R$-$D$ link (i.e., $P_2$ and $P_3$) when $P_1=0.2$ and $p=0.8$. In the figure, if the AoI difference is less than 0, the RP-AoR protocol has a lower average AoI. Otherwise, the average AoI of the SP-AoR protocol is lower. It can be observed that when $P_3$ is greater than about 0.45, the RP-AoR protocol always has a better AoI performance. The rationale behind this is similar to what we explained in the case of $(P_1,P_2,P_3)=(0.2,0.8,0.8)$ and $(0.2,0.3,0.8)$ in Fig. \ref{simulation_1}. When $P_1=0.2$ and $P_3>0.45$, it is difficult to have successful status update transmission through the $S$-$D$ link, while the $R$-$D$ link has a higher probability of successfully forwarding status updates to $D$. Therefore, in this case, the updates successfully delivered to $D$ by $R$ have a significant contribution on minimizing the average AoI. Since it is assumed that $p=0.8$ (i.e., the generation of status updates at $S$ is frequent), in the SP-AoR protocol, the updates at $R$ will always be preempted by $S$ and cannot be transmitted to $D$ to effectively reduce the AoI. However, the RP-AoR protocol will not be affected by the frequent status update generation, which results in a better AoI performance than the SP-AoR protocol. 


We also investigate the average AoI of both proposed AoR protocols against the transmission success probability of the $S$-$D$ link (i.e., $P_1$) when the considered system adopts the generate-at-will model (i.e., $p=1$). As depicted in Fig. \ref{simulation_3}, we consider two cases with $P_2=P_3=0.3$ and $P_2=P_3=0.8$, respectively. Based on Remark \ref{remark1}, the average AoI of the SP-AoR protocol is only related to $P_1$ in the generate-at-will model. Thus, there is only one average AoI curve for the SP-AoR protocol in Fig. \ref{simulation_3}. It can be observed that in both cases, the RP-AoR curve has an intersection with the SP-AoR curve. Specifically, when $(P_2,P_3)=(0.3,0.3)$ and $(P_2,P_3)=(0.8,0.8)$, the intersections occur at $P_1=0.1701$ and $P_1=0.4624$, respectively, which are consistent with the results calculated by using the formula in Remark \ref{remark_2}. This indicates that if the considered system has a new status update at the beginning of each time slot, we can quickly determine which of the SP-AoR protocol and the RP-AoR protocol is better for reducing the average AoI of the considered system by simply comparing the values of $P_1$ and $f(P_2,P_3)$.

\begin{figure}
	\centering
	\includegraphics[width=0.85\linewidth]{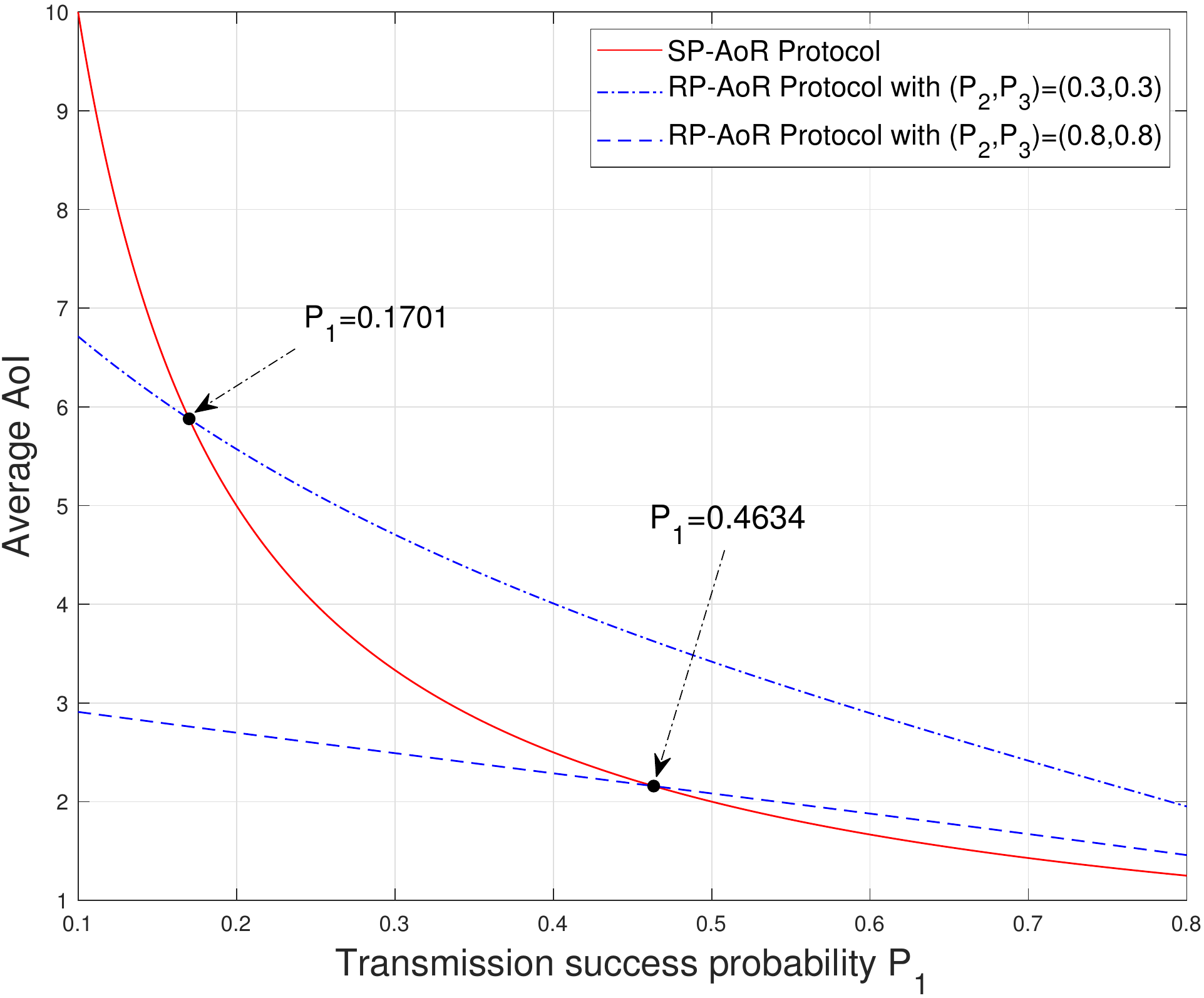}
	\caption{Average AoI of the SP-AoR protocol and the RP-AoR protocol versus the transmission success probability $P_1$ for different $(P_2, P_3)$ when the considered system adopts the generate-at-will model (i.e., $p=1$).}
	\label{simulation_3}
\end{figure}


\section{Conclusions}
In this paper, we considered a cooperative IoT system, where the source timely reports randomly generated status updates to the destination with the help of a relay. 
Considering the limited capabilities of IoT devices,
we devised two low-complexity AoR protocols, i.e., the SP-AoR protocol and the RP-AoR protocol, to reduce the AoI of the considered system. By analyzing the evolution of the instantaneous AoI, we derived the closed-form expressions of the average AoI for both proposed AoR protocols, which are functions of the generation probability of the status updates and the transmission success probability of each link. Based on the closed-form expressions, we further figured out the optimal generation probability of the status updates that minimizes the average AoI in both protocols. Simulation results validated our theoretical analysis, and demonstrated that the proposed protocols outperform each other under various system parameters. Note that based on the analytical results, we could quickly determine which protocol to apply for better AoI performance in the considered system. Furthermore, it was shown that the protocol with better performance can achieve near-optimal performance compared with the optimal scheduling policy solved by the MDP tool.

\begin{appendices}

\section{Proof of Theorem \ref{Theorem 2}} \label{prof3}

Recall that $\bar{\Delta}_{SP}=\frac{[1-(1-p)(1-P_{3})]\cdot [1-(1-p)(1-P_{1})(1-P_{2})]}{p[pP_{1}+(1-p)P_{3}-(1-p)(1-P_{1})(1-P_{2})P_{3}]}$. To proceed, we derive the first-order derivative of $\bar\Delta_{SP}$ with respect to (w.r.t.) $p$. After some algebra manipulations, we have
\begin{equation}\label{derisp}
\frac{\partial \bar{\Delta}_{SP}}{\partial p}=\frac{\mu \cdot p^2+ \lambda \cdot p+\xi}{p^2\!\cdot\!\big[pP_{1}\!+\!(1\!-\!p)P_{3}\!-\!(1\!-\!p)(1\!-\!P_{1})(1\!-\!P_{2}P_{3})\big]^2},
\end{equation}
where
\begin{subequations}
	\begin{equation}\label{proof3xi}
	\resizebox{.99\hsize}{!}{$\begin{aligned}
		\mu\!=
		\!-&(P_{1}^2P_{2}^2P_{3}^2\!-\!2P_{1}^2P_{2}P_{3}^2\!+\!2P_{1}^2P_{2}P_{3}\!-\!P_{1}^2P_{2}\!+\!P_{1}^2P_{3}^2\!
		-\!2P_{1}^2P_{3}\!\\	
		&\,+\!P_{1}^2\!
		-\!2P_{1}P_{2}^2P_{3}^2\!+\!2P_{1}P_{2}P_{3}^2\!-\!P_{1}P_{2}P_{3}\!+\!P_{1}P_{2}\!+\!P_{2}^2P_{3}^2\!-\!P_{2}P_{3}),
		\end{aligned}$}
	\end{equation}
	\begin{equation}\label{proof3psi}
	\resizebox{.91\hsize}{!}{$\begin{aligned}
		\lambda\!=\, &2P_{1}^2P_{2}^2P_{3}^2\!-\!4P_{1}^2P_{2}P_{3}^2\!+\!2P_{1}^2P_{2}P_{3}\!+\!2P_{1}^2P_{3}^2\!
		-\!2P_{1}^2P_{3}\!\\
		&\qquad\qquad\qquad-\!4P_{1}P_{2}^2P_{3}^2\!+\!4P_{1}P_{2}P_{3}^2\!-\!2P_{1}P_{2}P_{3}\!+\!2P_{2}^2P_{3}^2,
		\end{aligned}$}
	\end{equation}
	\begin{equation}\label{proof3omega}
	\resizebox{.99\hsize}{!}{$\begin{aligned}
		\xi\!=\!-(P_{1}^2P_{2}^2P_{3}^2\!-\!2P_{1}^2P_{2}P_{3}^2\!+\!P_{1}^2P_{3}^2
		\!-\!2P_{1}P_{2}^2P_{3}^2\!+\!2P_{1}P_{2}P_{3}^2\!+\!P_{2}^2P_{3}^2),
		\end{aligned}$}
	\end{equation}
\end{subequations}
are defined for the notation simplicity.

After a careful observation on the right hand side (RHS) of \eqref{derisp}, we can deduce that the sign of $\frac{\partial \bar{\Delta}_{SP}}{\partial p}$ is only determined by the numerator
\begin{equation}
\kappa(p) = \mu \!\cdot\! p^2+\lambda \!\cdot\! p+\xi,
\end{equation}
since the denominator is always large than zero. To determine the monotonicity of the function $\bar\Delta_{SP}$, we need to investigate the properties of the quadratic function $\kappa(p)$ on the feasible set of $p$ $\left(\mathrm{i.e.,}\  (0, 1]\right)$. 

Firstly, after some algebra manipulations, we note that
\begin{subequations}
	\begin{equation}
	\xi=-P_{3}^2\Big[P_{1}^2(P_{2}-1)^2+P_{2}\big(P_{2}+2P_{1}(1-P_{2})\big)\Big],
	\end{equation}
	\begin{equation}
	\lambda^2-4\mu \xi=4P_{3}^2P_{2}(P_{3}-P_{1})(1-P_{1})(P_{1}+P_{2}-P_{1}P_{2})^2.
	\end{equation}
\end{subequations}
As $0<P_{1}<P_{2}<1$ and $0<P_{1}<P_{3}<1$, we have that $\xi<0$ and $\lambda^2-4\mu \xi>0$. This means that the curve of $\kappa(p)$ always intersects the Y-axis at the negative half of the Y-axis and it always has two roots on the X-axis, which can be given by
\begin{subequations}
	\begin{equation}
	x_{1}=\frac{-\lambda -\sqrt{\lambda ^2-4\mu  \xi}}{2\mu},
	\end{equation}
	\begin{equation}\label{proof3x2}
	x_{2}=\frac{-\lambda +\sqrt{\lambda ^2-4\mu  \xi}}{2\mu},
	\end{equation}
\end{subequations}
where $x_{2}>x_{1}$. 

To further characterize the shapes of the function $\kappa(p)$, we now investigate $\mu$ and $\lambda$ for the following four possible cases:

1) When $\mu>0$ and $\lambda>0$, we always have $x_{1}<0$ and $x_{2}>0$. Since $p$ is in the range of $(0, 1]$, we have two sub-cases: (a) $x_{2}<1$ and (b) $x_{2}>1$. We draw the possible shapes of $\kappa(p)$ versus $p$ for the sub-case (a) and sub-case (b) in Fig. \ref{appendix1} (a) and Fig. \ref{appendix1} (b), respectively.

From Fig. \ref{appendix1} (a), we can see that in the sub-case (a), $\kappa(p)<0$ holds for $p\in(0, x_{2}]$, and $\kappa(p)>0$ holds for $p\in(x_{2}, 1]$. This means that the function $\bar\Delta$ is decreasing for $p\in(0, x_{2}]$ and is increasing for $p\in(x_{2}, 1]$. As a result, the minimum value of $\bar\Delta$ is achieved at $p=x_{2}$.

From Fig. \ref{appendix1} (b), we can see that in the sub-case (b), $\kappa(p)<0$ holds for $p\in(0, 1]$. Therefore, the function $\bar\Delta$ is always decreasing for $p\in(0, 1]$, which indicates that the value of $\bar\Delta$ is minimized at $p=1$.


2) When $\mu>0$ and $\lambda<0$, we always have $x_{1}<0$ and $x_{2}>0$. Since $p$ is in the range of $(0,1]$, we also have two sub-cases: (a) $x_{2}<1$ and (b) $x_{2}>1$. The possible shapes of $\kappa(p)$ versus $p$ for the sub-case (a) and sub-case (b) are depicted as Fig. \ref{appendix2} (a) and Fig. \ref{appendix2} (b), respectively. 

Similar to case 1), we can prove that the minimum value of $\bar\Delta$ is achieved by setting $p=x_{2}$ and $p=1$ for the case $x_{2}<1$ and $x_{2}>1$, respectively.

3) When $\mu<0$ and $\lambda<0$, we always have $x_{1}<x_{2}<0$. Therefore, there only exists one possible shape of $\kappa(p)$ versus $p$ in this case, which is depicted as Fig. \ref{appendix3}. 

From Fig. \ref{appendix3}, it is straightforward to see that $\kappa(0)<0$ holds for $p\in (0, 1]$. This means that the function $\bar\Delta$ is always decreasing for $p\in (0, 1]$ in this case. As a result, the minimum value of $\bar\Delta$ is achieved at $p=1$.

4) When $\mu<0$ and $\lambda>0$, we always have $x_{2}>x_{1}>0$. Therefore, the possible shapes of $\kappa(p)$ versus $p$ can be depicted as Fig. \ref{appendix4}. According to \eqref{proof3psi}, after some algebra manipulations, $\lambda$ can be rewritten as 
\begin{equation}
\lambda\!=-2P_{3}(P_{1}\!+\!P_{2}\!-\!P_{1}P_{2})(P_{1}\!-\!P_{1}P_{3}\!-\!P_{2}P_{3}\!+\!P_{1}P_{2}P_{3}).
\end{equation}
As $\lambda>0$, $0<P_{1}<P_{2}<1$ and $0<P_{1}<P_{3}<1$, we thus have $P_{1}-P_{1}P_{3}-P_{2}P_{3}+P_{1}P_{2}P_{3}<0$, which can be rewritten as
\begin{equation}\label{proof3P1}
P_{1}<\frac{P_{2}P_{3}}{1-P_{3}+P_{2}P_{3}}. 
\end{equation}

To further validate this case, we investigate the value of $-\mu$ under the constraint of \eqref{proof3P1}. We rewrite $-\mu$ as a quadratic function of $P_{1}$, which can be expressed as 
\begin{equation}
-\mu(P_{1})=\psi P_{1}^2+\chi P_{1}+\omega,
\end{equation}
where
\begin{subequations}
	\begin{equation}
	\psi = P_{2}^2P_{3}^2-2P_{2}P_{3}^2+2P_{2}P_{3}-P_{2}+P_{3}^2-2P_{3}+1,
	\end{equation}
	\begin{equation}\label{proof3miu}
	\chi=2P_{2}P_{3}^2(1-P_{2})+P_{2}(1-P_{3}),
	\end{equation}
	\begin{equation}\label{proof3unknown}
	\omega= P_{2}P_{3}(P_{2}P_{3}-1),
	\end{equation}
\end{subequations}
are defined for the simplicity of notations.

As $0<P_{1}<P_{2}<1$ and $0<P_{1}<P_{3}<1$, according to \eqref{proof3miu} and \eqref{proof3unknown}, we have $\chi>0$ and $\omega<0$. In addition, if we set $P_{1}=1$, after some manipulations, we can attain that $-\mu({P_{1}})=(P_{3}-1)^2>0$. Similarly, if we substitute $P_{1}=\frac{P_{2}P_{3}}{1-P_{3}+P_{2}P_{3}}$ into the expression of $-\mu(P_{1})$, we can have $-\mu (P_{1})=\frac{P_{2}P_{3}(P_{2}-1)(P_{3}-1)^2}{(P_{2}P_{3}-P_{3}+1)^2}<0$. Based on the above analysis, we can draw the possible shapes of $-\mu(P_{1})$ versus $P_{1}$ for $\psi>0$ and $\psi<0$ in Fig. \ref{appendix5} (a) and Fig. \ref{appendix5} (b), respectively. Note that in these figures, we also show the possible positions of the point with X-coordinate being $\frac{P_{2}P_{3}}{1-P_{3}+P_{2}P_{3}}$.

From Fig. \ref{appendix5} (a) and Fig. \ref{appendix5} (b), we can see that if $P_{1}<\frac{P_{2}P_{3}}{1-P_{3}+P_{2}P_{3}}$ and $0<P_{1}<1$,  $-\mu<0$ always holds. This means that we always have $\mu>0$ if $\lambda>0$. This is contrary to the initial assumptions of this case (i.e., $\mu<0$ and $\lambda>0$). Therefore, the case 4) is not valid.

In conclusion, we can find that in all of the possible case 1), case 2) and case 3), the minimum value of $\bar\Delta$ is achieved by setting $p=x_{2}$ and $p=1$ when $x_{2}<1$ and $x_{2}>1$, respectively. Therefore, based on \eqref{proof3xi}, \eqref{proof3psi}, \eqref{proof3omega} and \eqref{proof3x2}, after some algebra manipulations, the optimal $p$ that minimizes $\bar\Delta$ can be given by \eqref{psp}. This completes the proof.

\begin{figure}[htbp]
	\centering
	\subfigure[When $x_{2}<1$]{\includegraphics[height=2.6cm]{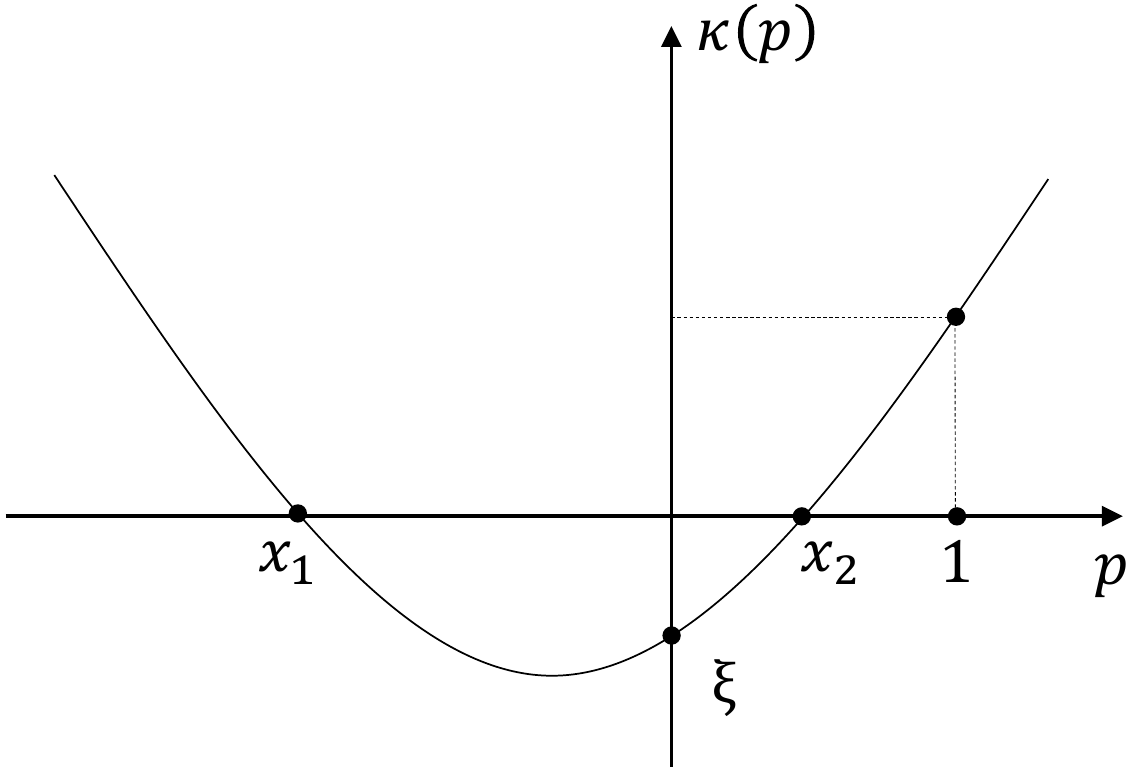}}
	\qquad
	\subfigure[When $x_{2}>1$]{\includegraphics[height=2.6cm]{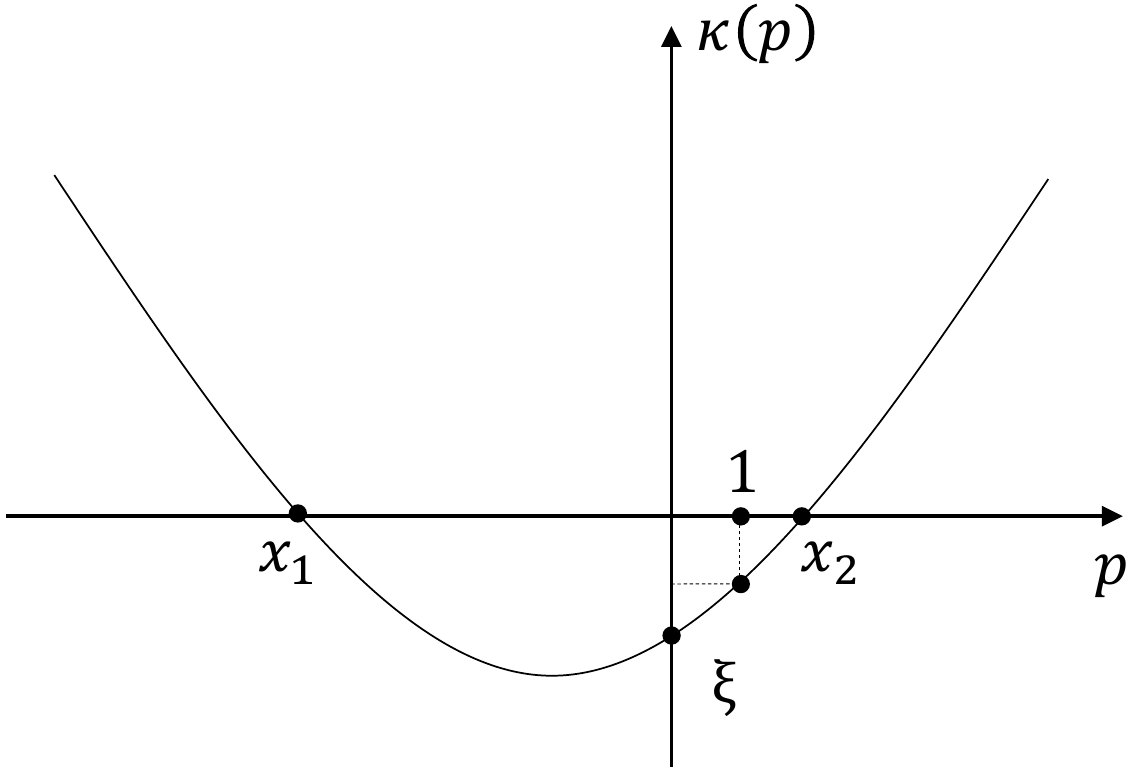}}
	\caption{Possible shapes for the function $\kappa(p)$ versus $p$ when $\mu>0$ and $\lambda>0$.}
	\label{appendix1}
\end{figure}

\begin{figure}[htbp]
	\centering
	\subfigure[When $x_{2}<1$]{\includegraphics[height=2.6cm]{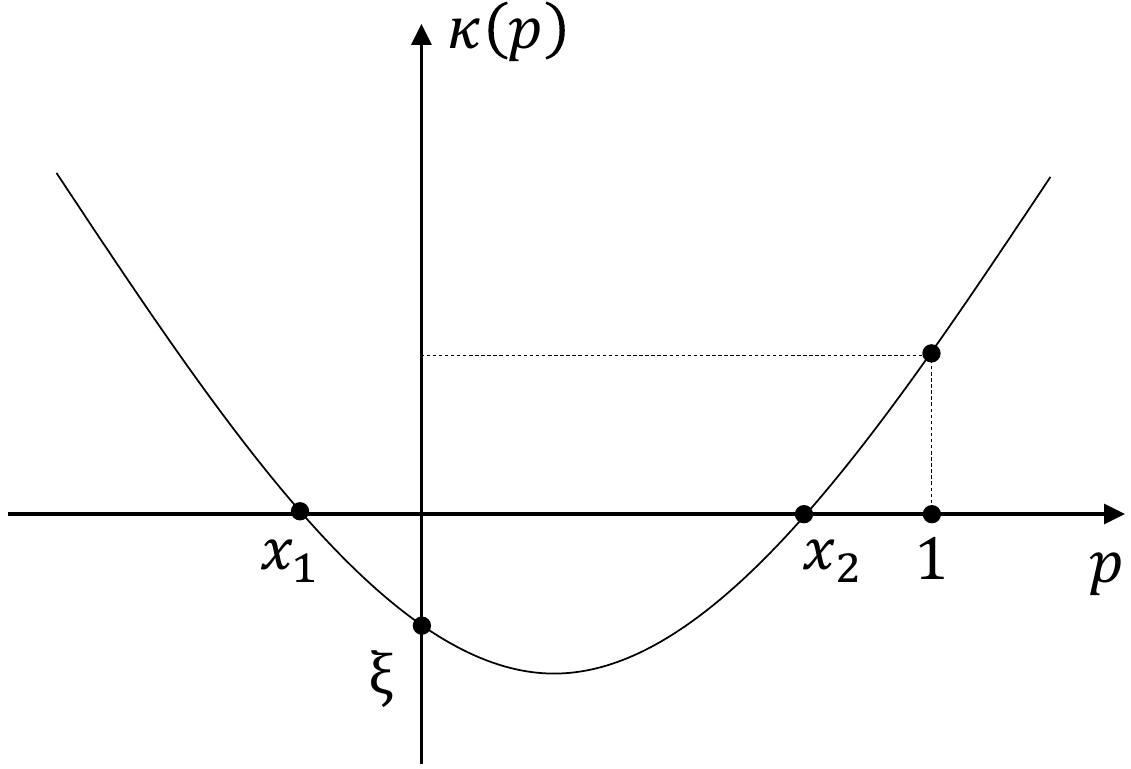}}
	\qquad
	\subfigure[When $x_{2}>1$]{\includegraphics[height=2.6cm]{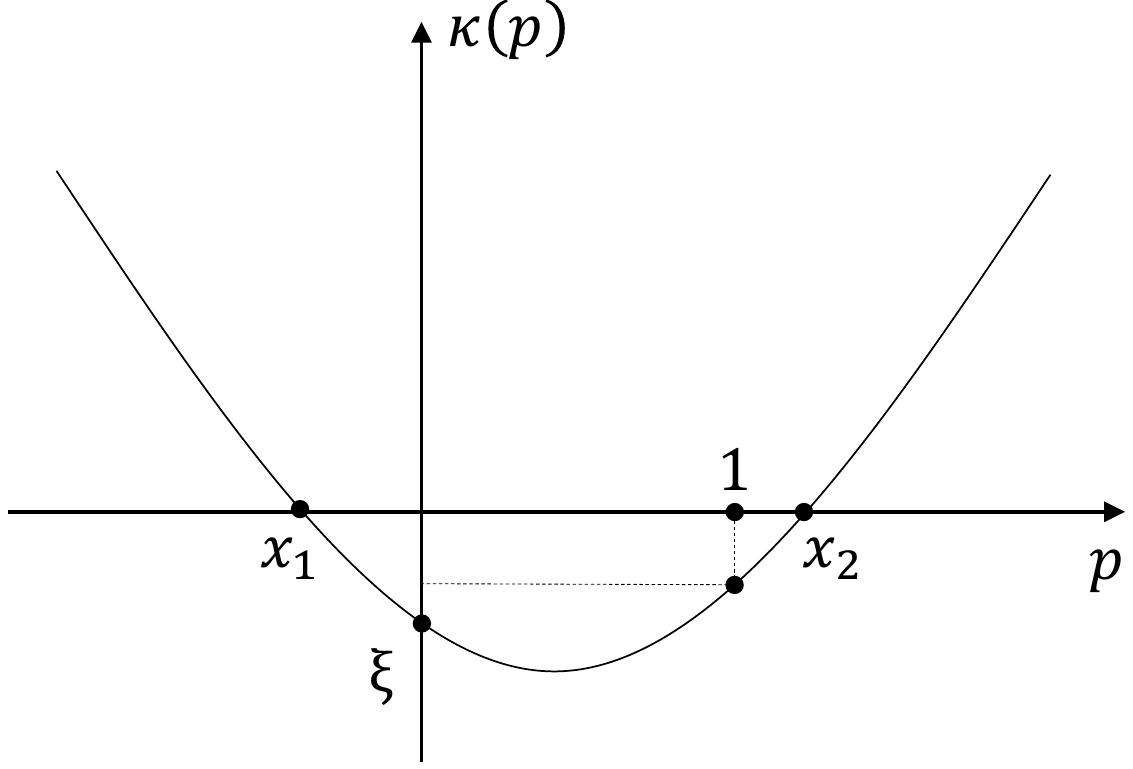}}
	\caption{Possible shapes for the function $\kappa(p)$ versus $p$ when $\mu>0$ and $\lambda<0$.}
	\label{appendix2}
\end{figure}

\begin{figure}[tbp]
	\centering{\includegraphics[height=2.8cm]{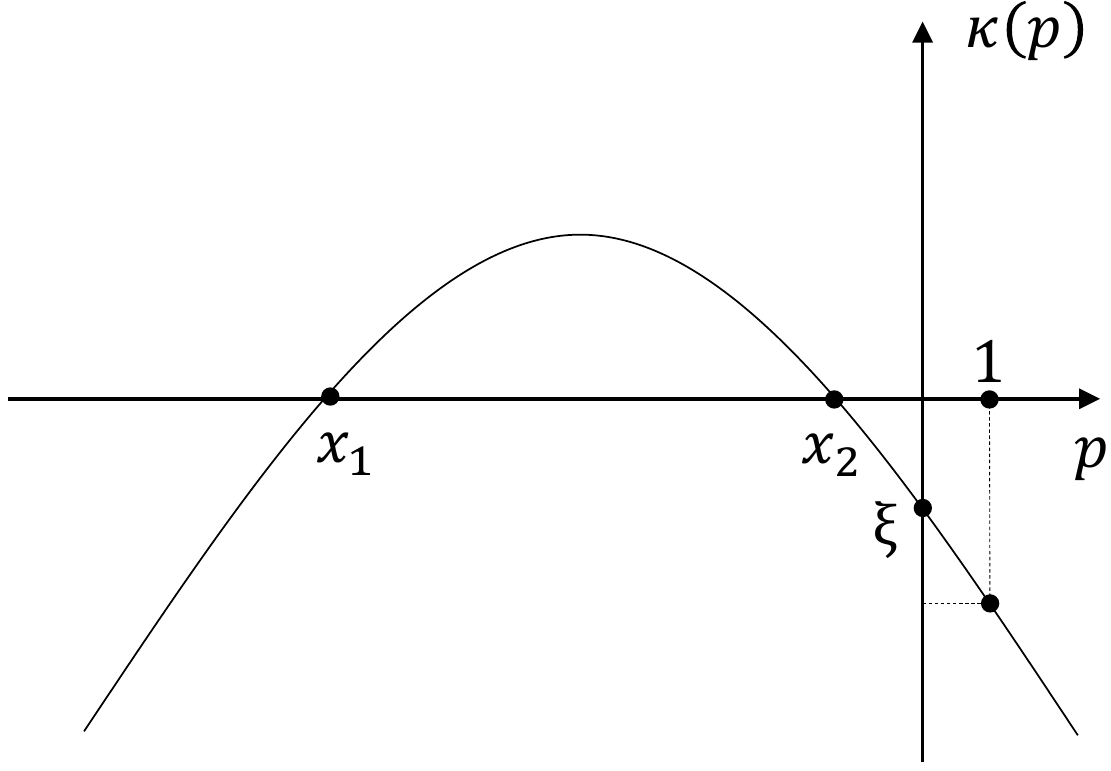}}
	\vspace{-0.7 em}
	\caption{Possible shapes for the function $\kappa(p)$ versus $p$ when $\mu<0$ and $\lambda<0$.}
	\label{appendix3}
\end{figure}

\begin{figure}[htbp]
	\centering
	\subfigure[When $x_{2}<1$]{\includegraphics[height=2.6cm]{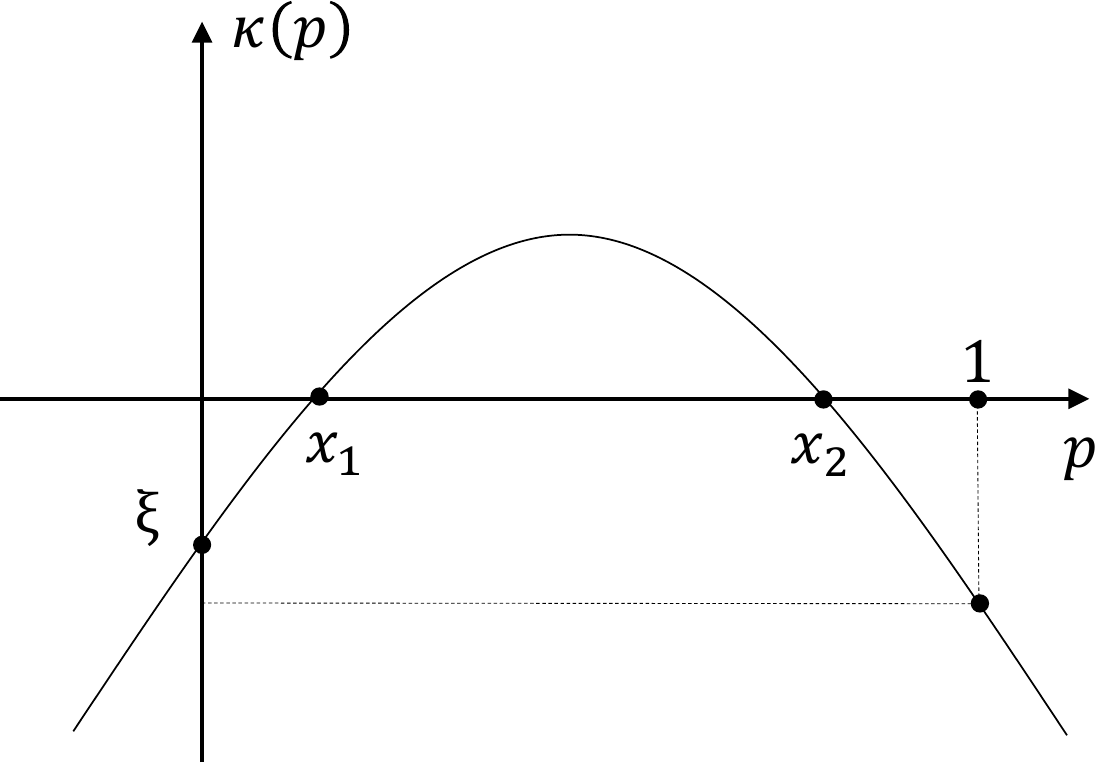}}
	\qquad
	\subfigure[When $x_{2}>1$]{\includegraphics[height=2.6cm]{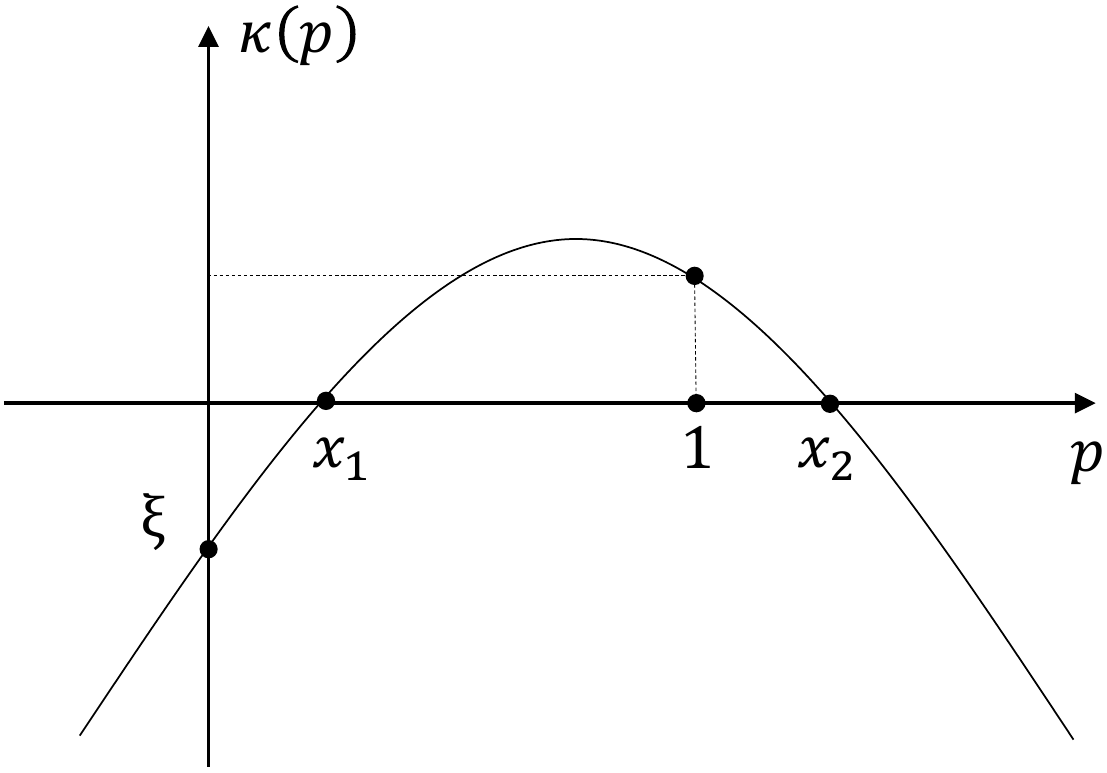}}
	\caption{Possible shapes for the function $\kappa(p)$ versus $p$ when $\mu<0$ and $\lambda>0$.}
	\label{appendix4}
\end{figure}

\begin{figure}[htbp]
	\centering
	\subfigure[When $\psi>0$]{\includegraphics[height=2.6cm]{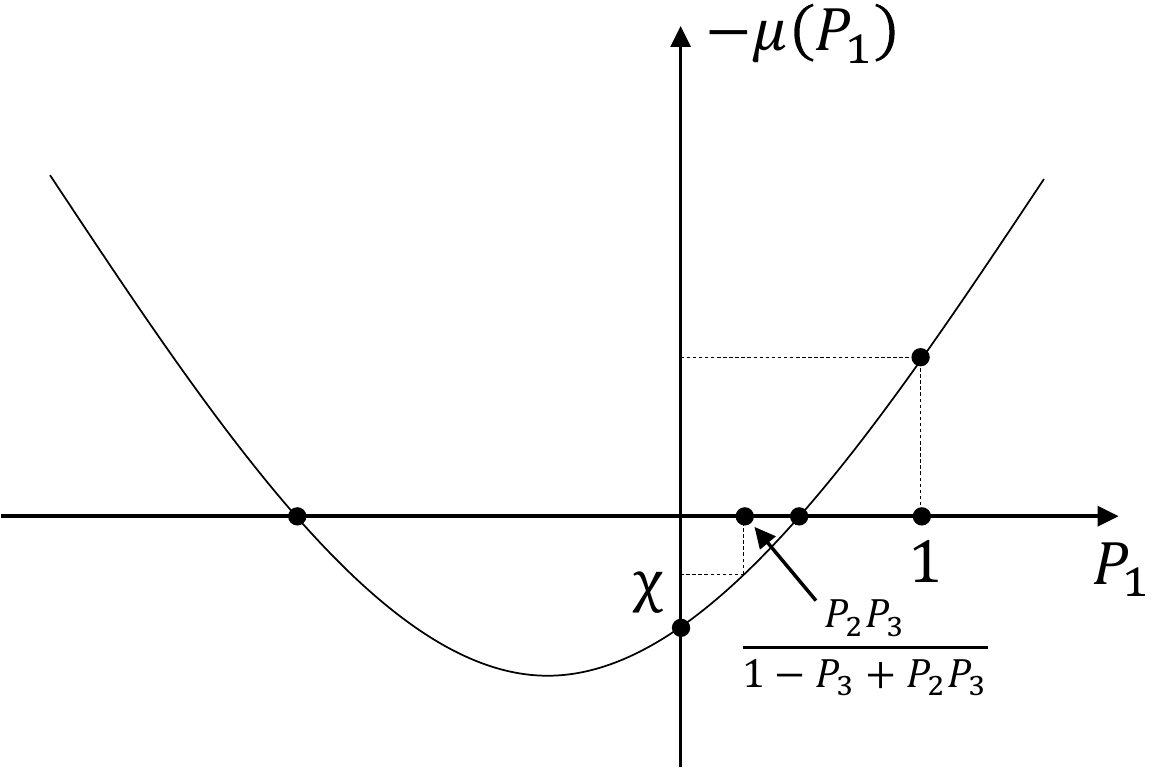}}
	\qquad
	\subfigure[When $\psi<0$]{\includegraphics[height=2.6cm]{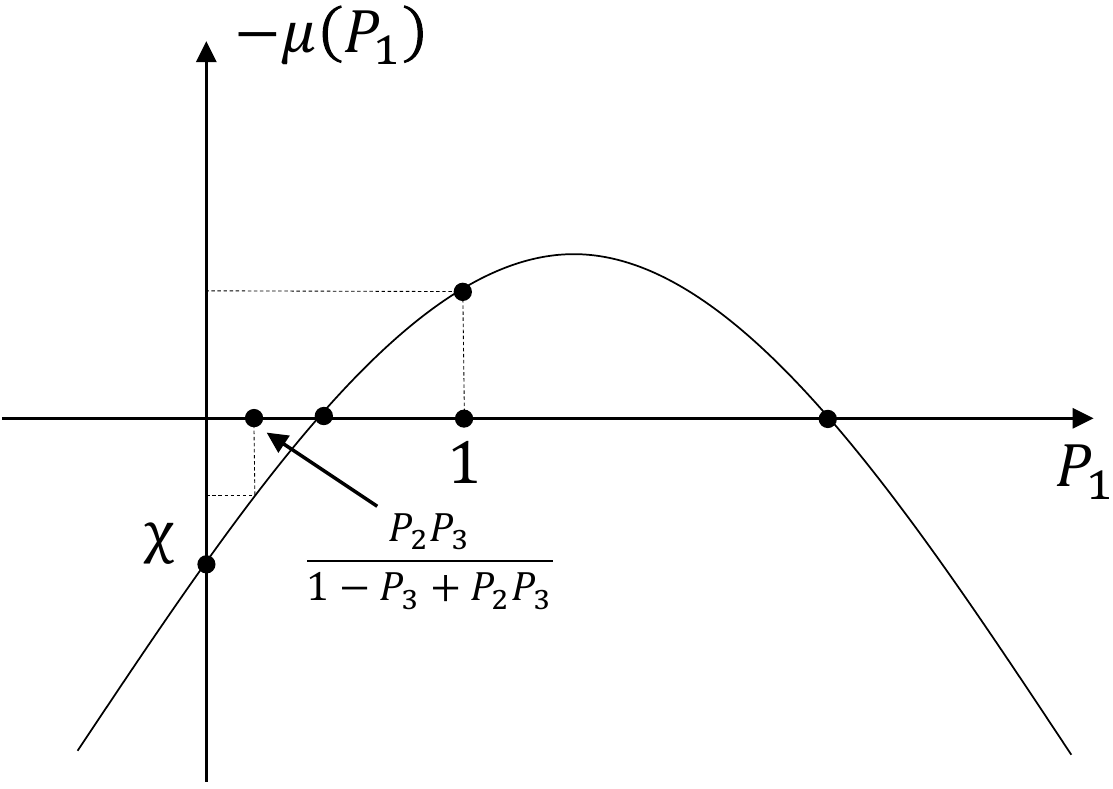}}
	\caption{Possible shapes for the function $-\mu(P_{1})$ versus $P_{1}$.}
	\label{appendix5}
\end{figure}

\end{appendices}

\bibliography{References}
\bibliographystyle{IEEEtran}

\end{document}